\documentclass[12pt,british,english]{article}
\usepackage[T1]{fontenc}
\usepackage[latin9]{inputenc}
\usepackage{geometry}
\usepackage{xcolor}
\usepackage{babel}
\usepackage{tipa}
\usepackage{amsmath}
\usepackage{amsthm}
\usepackage{amssymb}
\usepackage{setspace}
\usepackage[authoryear]{natbib}
\usepackage[unicode=true,pdfusetitle,
 bookmarks=true,bookmarksnumbered=false,bookmarksopen=false,
 breaklinks=false,pdfborder={0 0 0},pdfborderstyle={},backref=false,colorlinks=true]
 {hyperref}
\hypersetup{
 citecolor=blue, linkcolor=blue, urlcolor=blue}

\makeatletter
\theoremstyle{remark}
\newtheorem{rem}{\protect\remarkname}
\theoremstyle{definition}
\newtheorem{defn}{\protect\definitionname}
\theoremstyle{plain}
\newtheorem{thm}{\protect\theoremname}
\theoremstyle{plain}
\newtheorem{cor}{\protect\corollaryname}
\theoremstyle{plain}
\newtheorem{lem}{\protect\lemmaname}
\theoremstyle{plain}
\newtheorem{prop}{\protect\propositionname}
\theoremstyle{definition}
 \newtheorem{example}{\protect\examplename}

\usepackage{lmodern}
\usepackage[T1]{fontenc}

\makeatother

\addto\captionsbritish{\renewcommand{\corollaryname}{Corollary}}
\addto\captionsbritish{\renewcommand{\definitionname}{Definition}}
\addto\captionsbritish{\renewcommand{\examplename}{Example}}
\addto\captionsbritish{\renewcommand{\lemmaname}{Lemma}}
\addto\captionsbritish{\renewcommand{\propositionname}{Proposition}}
\addto\captionsbritish{\renewcommand{\remarkname}{Remark}}
\addto\captionsbritish{\renewcommand{\theoremname}{Theorem}}
\addto\captionsenglish{\renewcommand{\corollaryname}{Corollary}}
\addto\captionsenglish{\renewcommand{\definitionname}{Definition}}
\addto\captionsenglish{\renewcommand{\examplename}{Example}}
\addto\captionsenglish{\renewcommand{\lemmaname}{Lemma}}
\addto\captionsenglish{\renewcommand{\propositionname}{Proposition}}
\addto\captionsenglish{\renewcommand{\remarkname}{Remark}}
\addto\captionsenglish{\renewcommand{\theoremname}{Theorem}}
\providecommand{\corollaryname}{Corollary}
\providecommand{\definitionname}{Definition}
\providecommand{\examplename}{Example}
\providecommand{\lemmaname}{Lemma}
\providecommand{\propositionname}{Proposition}
\providecommand{\remarkname}{Remark}
\providecommand{\theoremname}{Theorem}

\begin{document}
\title{Coevolution of Deception and Preferences: \\
Darwin and Nash Meet Machiavelli\thanks{Valuable comments were provided by the anonymous associate editor
and referees, Vince Crawford, Eddie Dekel, Jeffrey Ely, Itzhak Gilboa,
Christoph Kuzmics, Larry Samuelson, Jörgen Weibull, and Okan Yilankaya,
as well as participants at presentations at Oxford University, Queen
Mary University, G.I.R.L.13 in Lund, the Toulouse Economics and Biology
Workshop, DGL13 in Stockholm, the 25th International Conference on
Game Theory at Stony Brook, and the Biological Basis of Preference
and Strategic Behaviour 2015 conference at Simon Fraser University.
Yuval Heller is grateful to the \emph{European Research Council} for
its financial support (starting grant \#677057). Erik Mohlin is grateful
to \emph{Handelsbankens forskningsstiftelser} (grant \#P2016-0079:1)
and the \emph{Swedish Research Council} (grant \#2015-01751) for its
financial support.}}
\author{Yuval Heller\thanks{Affiliation: Department of Economics, Bar Ilan University. Address:
Ramat Gan 5290002, Israel. E-mail: yuval.heller@biu.ac.il. } and Erik Mohlin\thanks{Affiliation: Department of Economics, Lund University. Address: Tycho
Brahes väg 1, 220 07 Lund, Sweden. E-mail: erik.mohlin@nek.lu.se.}}

\maketitle
Final pre-print of a paper published in \emph{Games and Economic Behavior},
113, 2019, pp. 223-247.
\begin{abstract}
We develop a framework in which individuals' preferences coevolve
with their abilities to deceive others about their preferences and
intentions. Specifically, individuals are characterised by (i) a level
of cognitive sophistication and (ii) a subjective utility function.
Increased cognition is costly, but higher-level individuals have the
advantage of being able to deceive lower-level opponents about their
preferences and intentions in some of the matches. In the remaining
matches, the individuals observe each other's preferences. Our main
result shows that, essentially, only efficient outcomes can be stable.
Moreover, under additional mild assumptions, we show that an efficient
outcome is stable if and only if the gain from unilateral deviation
is smaller than the effective cost of deception in the environment. 

Preprint of the 
\end{abstract}

\section{Introduction\label{sec:Introduction}}

For a long time economists took preferences as given. The study of
their origin and formation was considered a question outside the scope
of economics. Over the past two decades this has changed dramatically.
In particular, there is now a large literature on the evolutionary
foundations of preferences (for an overview, see \citealp{Robson_Samuelson_2011}).
A prominent strand of this literature is the so-called ``indirect
evolutionary approach,\textquotedblright{} pioneered by \citet{Guth_Yaari1992}
(term coined by \citealp{Guth_1995_named_ind_evo_approach}). This
approach has been used to explain the existence of a variety of ``non-standard\textquotedblright{}
preferences that do not coincide with material payoffs, e.g., altruism,
spite, and reciprocal preferences.\footnote{For example, \citet{Bester_Guth1998}, \citet{Bolle_2000}, and \citet{Possajennikov2000}
study combinations of altruism, spite, and selfishness. \citet{Ellingsen_1997}
finds that preferences that induce aggressive bargaining can survive
in a Nash demand game. \citet{Fershtman_Weiss_1998} study evolution
of concerns for social status. \citet{Sethi_Somanthan2001} study
the evolution of reciprocity in the form of preferences that are conditional
on the opponent's preference type. In the context of the finitely
repeated Prisoner's Dilemma, \citet{Guttman_2003} explores the stability
of conditional cooperation. \citet{Dufwenberg_Guth_1999_EvoandDel}
study firm's preferences for large sales. \citet{Guth_Napel2006}
study preference evolution when players use the same preferences in
both ultimatum and dictator games. \citet{Kockcesen_Ok_2000} investigate
survival of more general interdependent preferences in aggregative
games. \citet{Friedman_Singh_2009} show that vengefulness may survive
if observation has some degree of informativeness. Recently, \citet{Norman_2012_DynPrefEvo}
has shown how to adapt some of these results into a dynamic model} Typically, the non-materialistic preferences in question convey some
form of commitment advantage that induces opponents to behave in a
way that benefits individuals with non-materialistic preferences,
as described by \citet{Schelling_1960_strategyofconflict} and \citet{Frank_1987_if}.
Indeed, \citet{Heifetz_et_al_2007_whattomax} show that this kind
of result is generic.

A crucial feature of the indirect evolutionary approach is that preferences
are explicitly or implicitly assumed to be at least partially observable.\footnote{\citet{Gamba_2013} is an interesting exception. She assumes play
of a self-confirming equilibrium, rather than a Nash equilibrium,
in an extensive-form game. This allows for evolution of non-materialistic
preferences even when they are completely unobservable. An alternative
is to allow for a dynamic that is not strictly payoff monotonic. This
approach is pursued by \citet{frenkel2012endowment}, who show that
multiple biases (inducing non-materialistic preferences) can survive
in non-monotonic evolutionary dynamics even if they are unobservable,
because each approximately compensates for the errors of the others.} Consequently the results are vulnerable to the existence of mimics
who signal that they have, say, a preference for cooperation, but
actually defect on cooperators, thereby earning the benefits of having
the non-standard preference without having to pay the cost (\citealp{Samuelson_2001_intro}).
The effect of varying the degree to which preferences can be observed
has been investigated by \citet{Ok_VegaRedondo_2001}, \citet{Ely_Yilankaya_2001},
\citet{Dekel_Ely_Yilankaya_2006}, and \citet{Herold_Kuzmics_2009}.
They confirm that the degree to which preferences are observed decisively
influences the outcome of preference evolution. 

Yet, the degree to which preferences are observed is still exogenous
in these models. In reality we would expect both the preferences and
the ability to observe or conceal them to be the product of an evolutionary
process.\footnote{On this topic, \citet{Robson_Samuelson_2011} write: ``The standard
argument is that we can observe preferences because people give signals
\textendash{} a tightening of the lips or flash of the eyes \textendash{}
that provide clues as to their feelings. However, the emission of
such signals and their correlation with the attendant emotions are
themselves the product of evolution. {[}...{]} We cannot simply assume
that mimicry is impossible, as we have ample evidence of mimicry from
the animal world, as well as experience with humans who make their
way by misleading others as to their feelings, intentions and preferences.
{[}...{]} In our view,\emph{ the indirect evolutionary approach will
remain incomplete until the evolution of preferences, the evolution
of signals about preferences, and the evolution of reactions to these
signals, are all analysed within the model}.\textquotedblright{} {[}Emphasis
added{]} (pp. 14\textendash 15)} \emph{This paper provides a first step towards filling in the missing
link between evolution of preferences and evolution of how preferences
are concealed, feigned, and detected.}\footnote{The recent working paper of \citet{gauer2016cognitive} presents a
different way to endogenising the observability of preferences. Specifically,
they assume that preferences are ex ante uncertain, and that each
player may exert a cognitive effort to privately observe the opponent's
preferences. }\emph{ }In our model the ability to observe preferences and the ability
to deceive and induce false beliefs about preferences are endogenously
determined by evolution, jointly with the evolution of preferences.
Cognitively more sophisticated players have positive probability of
deceiving cognitively less sophisticated players. Mutual observation
of preferences occurs only in matches in which such deception fails.
This setup is general enough to encompass both the standard indirect
evolutionary model where preferences are always observed, and the
reverse case in which more sophisticated types always deceive lower
types, as well as all intermediate cases between these two extremes.
\emph{We} \emph{find that, generically, only efficient outcomes can
be played in stable population states.} Moreover, we define a single
number that captures the effective cost of deception against naive
opponents, and show that \emph{an efficient outcome is stable if and
only if the gain from a unilateral deviation is smaller than the effective
cost of deception}. 

\paragraph{Overview of the Model.}

As is common in standard evolutionary game theory we assume an infinite
population of individuals who are uniformly randomly matched to play
a symmetric normal form game.\footnote{It is known that positive assortative matching is conducive to the
evolution of altruistic behaviour (\citealp{Hines_et_al_1979}) and
non-materialistic preferences even when preferences are perfectly
unobservable (\citealp{Alger_Weibull_HomoMoralis}; \citealp{Bergstrom_1995}).
It is also known that finite populations allow for evolution of spiteful
behaviours (\citealp{Schaffer_1988}) and non-materialistic preferences
(\citealp{Huck_Oechssler1999}). By assuming that individuals are
uniformly randomly matched in an infinite population, we avoid confounding
these effects with the effect of endogenising the degree of observability.} Each individual has a type, which is a tuple, consisting of a \emph{preference
component} and a \emph{cognitive component}. The preference component
is identified with a subjective utility function over the set of outcomes
(i.e. action profiles), which may differ from the objective payoffs
(i.e., fitness) of the underlying game. The cognitive component is
simply a natural number representing the level of cognitive sophistication
of the individual.\footnote{The one-dimensional representation of cognitive ability reflects the
idea that if one is good at deceiving others, then one is more likely
to be good also at reading others and avoiding being deceived by them.
In this paper we simplify this relation by assuming a perfect correlation
between the two abilities, and leave the study of more general relations
for future research.}$^{,}$\footnote{Remark \ref{enu:An-alternative-interpretation} in Section \ref{subsec:Configurations}
presents an alternative interpretation of our model, according to
which this cognitive component represents the agent's social status,
rather than the agent's ability to deceive other agents.} The cost of increased cognition is strictly positive. 

When two individuals with different cognitive levels are matched,
there is positive probability (which may depend on the cognitive levels
of both agents) that the agent with the higher level deceives his
opponent. For the sake of tractability, and in order not to limit
the degree to which higher levels may exploit lower levels, we model
a strong form of deception. The deceiver observes the opponent's preferences
perfectly, and is allowed to choose whatever she wants the deceived
party to believe about the deceiver's intended action choice. A strategy
profile that is consistent with this form of deception is called a
\emph{deception equilibrium}. With the remaining probability (or with
probability one if both agents have the same cognitive level) there
is no deception in the match. In this case, we assume that each player
observes the opponent's preferences, and the individuals play a Nash
equilibrium of the complete information game induced by their subjective
preferences.

The state of a population is described by a \emph{configuration},
consisting of a type distribution and a behaviour policy. The \emph{type
distribution} is simply a finite support distribution on the set of
types. The \emph{behaviour policy} specifies a Nash equilibrium for
each match without deception, and a deception equilibrium for each
match with deception. In a \emph{neutrally stable configuration} all
incumbents earn the same expected fitness, and if a small group of
mutants enter they earn weakly less than the incumbents in any \emph{focal}
post-entry state. A focal post-entry state is one in which the incumbents
behave against each other in the same way as before the mutants entered.

\paragraph{Main Results.}

We say that a strategy profile is (fitness-)efficient if it maximises
the sum of objective payoffs. Theorem \ref{thm:highest-type-behaivor}
shows that in any stable configuration, any type $\bar{\theta}$ with
the highest cognitive level in the incumbent population must play
an efficient strategy profile when meeting itself. The intuition is
that otherwise a highest-type mutant who mimics the play of $\bar{\theta}$
against all incumbents while playing an efficient strategy profile
against itself would outperform type $\bar{\theta}$ (a novel application
of the ``secret handshake'' argument due to \citealp{robson1990efficiency}). 

Next we restrict attention to generic games (i.e. games that result
with probability one if fitness payoffs are independently drawn from
a continuous distribution) and obtain our first main result: \emph{any
stable configuration must induce efficient play} in all matches between
all types. The idea of the proof can be briefly sketched as follows.
We first show that any type $\theta$ in a stable configuration must
play an efficient strategy profile when meeting \textit{itself}. Otherwise
a mutant who has the same level as $\theta$ and the same utility
function as $\theta$, but who plays efficiently against itself, could
invade the population. Next, we show that \textit{any} two types must
play an efficient strategy profile. The intuition is that otherwise
the average within-group fitness would be higher than the between-group
fitness, which implies instability in the face of small perturbations
in the frequency of the types: a type who became slightly more frequent
would have a higher fitness than the other incumbents, and this would
move the population away from the original configuration.

The existing literature (e.g., \citealp{Dekel_Ely_Yilankaya_2006})
has demonstrated that if players perfectly observe each other's preferences
(or do so with sufficiently high probability), then only efficient
outcomes are stable. As was pointed out above, our model encompasses
the limiting case in which it is arbitrarily ``cheap and easy''
to deceive the opponent, i.e. the case in which the marginal cost
of an additional cognitive level is very low, and having a slightly
higher cognitive level allows a player to deceive the opponent with
probability one. A key contribution of the paper is to show that even
when it is cheap and easy to deceive the opponent, \emph{the seemingly
mild assumption of perfect observability, and Nash equilibrium behaviour,
among players with the same cognitive level is enough to ensure that}
\emph{stability implies efficiency}. 

In order to obtain sufficient conditions for stability we restrict
attention to generic games that admit a ``punishment action'' that
ensures that the opponent achieves strictly less than the symmetric
efficient fitness payoff. For games satisfying this relatively mild
requirement we fully characterise stable configurations. We define
the \emph{(fitness) deviation gain} of an action profile to be the
maximal fitness increase a player may obtain by unilaterally deviating
from this action profile (this gain is zero if and only if the action
profile is a Nash equilibrium of the underlying game). Next we define
the \emph{effective cost of deception} in the environment as the minimal
ratio between the cost of an increased cognitive level and the probability
that an agent with this level deceives an opponent with the lowest
cognitive level. Our second main result shows that an efficient action
profile is the outcome of a stable configuration if and only if its
deviation gain is smaller than the effective cost of deception. In
particular, \emph{efficient Nash equilibria are stable in all environments,
while non-Nash efficient action profiles are stable only as long as
the gain from a unilateral deviation is sufficiently small}. 

Next, we note that non-generic games may admit different kinds of
stable configurations. One particularly interesting family of non-generic
games is the family of zero-sum games, such as the Rock-Paper-Scissors
game. We analyse this game and characterise a heterogeneous stable
population (inspired by a related construction in \citealp{Conlisk_2001})
in which different cognitive levels coexist, players with equal levels
play the Nash equilibrium of the underlying game, and players with
higher levels beat their opponents but this gain is offset by higher
cognitive costs. 

Finally, in Section \ref{sec:Extensions} we discuss two extensions
of the model (which are formally analysed in Appendices \ref{Sect Interdependent}
and \ref{sec:partial-observability}): (1) we relax the assumption
that each agent perfectly observes the partner's preferences in matches
without deception, and (2) we allow for type-interdependent preferences
($\grave{\textrm{a}}$ la \citealp{Herold_Kuzmics_2009}), which are
represented by utility functions that are defined over both action
profiles and the opponent's type. 

\paragraph{Further Related Literature.\label{sec:Related-Literature}}

Our model is related to work in biology and evolutionary psychology
on the evolution of the ``theory of mind'' (\citealp{Premack_Wodruff_1978}),
specifically, the ``Machiavellian intelligence\textquotedblright{}
hypothesis (\citealp{Humphrey_1976}) and the ``social brain\textquotedblright{}
hypothesis (\citealp{Byrne_Whiten_1998}), according to which the
extraordinary cognitive abilities of humans evolved as a result of
the demands of social interactions, rather than the demands of the
natural environment: in a single-person decision problem there is
a fixed benefit from being smart, but in a strategic situation it
may be important to be smarter than the opponent. From an evolutionary
perspective, there is a trade-off between the benefit of outsmarting
the opponent and the non-negligible costs associated with increased
cognitive capacity (\citealp{Holloway_1996}; \citealp{Kinderman_etal_1998}).
Our model incorporates these features. 

There is a smaller literature on the evolution of strategic sophistication
within game theory; see, e.g., \citet{Stahl_93}, \citet{Banerjee_Weibull_1995},
\citet{Stennek_2000}, \citet{Conlisk_2001}, \citet{Abreu_Sethi_2003},
\citet{Mohlin_2012_ETM}, \citet{Rtischev_2016}, and \citet{heller2014three}.
Following these papers, we provide results to the effect that different
degrees of cognitive sophistication may coexist. 

\citet{robalino2016evolution} construct a model to demonstrate the
advantage of having a theory of mind (understood as an ability to
ascribe stable preferences to other players) over learning by reinforcement.
In novel games the ascribed preferences allow the agents with a theory
of mind to draw on past experience whereas a reinforcement learner
without such a model has to start over again. \citet{Hopkins_2014_compalt}
explains why costly signaling of altruism may be especially valuable
for those agents who have a theory of mind. 

\citet{robson1990efficiency} initiated a literature on evolution
in cheap-talk games by formulating the secret handshake effect: evolution
selects an efficient stable state if mutants can send messages that
the incumbents either do not see or do not benefit from seeing. Against
the incumbents a mutant plays the same action as the incumbents do,
but against other mutants the mutant plays an action that is a component
of the efficient equilibrium. Thus the mutants are able to invade
unless the incumbents are already playing efficiently. See also the
related analysis in \citet{Matsui_1991} and \citet{wiseman2001cooperation}.
We allow for deception and still find that efficiency is necessary
(though no longer sufficient) for stability. As pointed out by \citet{Warneryd_1991}
and \citet{Schlag_1993_cheapWP}, among others, problems arise if
either the incumbents use all available messages (so that there is
no message left for the incumbents to coordinate on) or the incumbents
follow a strategy that induces the mutants to play an action that
lowers the mutants' payoffs below those of the incumbents. To circumvent
this problem, \citet{Kim_Sobel_1995} use stochastic stability arguments
and \citet{Warneryd_1998_commcomplex} uses complexity costs. Similarly,
evolution selects an efficient outcome in our model, where the preferences
also serve the function of messages.

We conclude by mentioning three other related strands of literature
in which deception has been implicitly studied: (1) the ``strategic
teaching'' literature, which studies situations in which sophisticated
agents manipulate the learning input of opponents in order to change
the beliefs and future actions of these opponents (see, e.g., \citealp[Section 8.11]{fudenberg1998theory,camerer2002sophisticated,schipper2017strategic});
(2) the ``reputation'' literature, in which a long-run player manipulates
the beliefs and behaviour of short-run opponents (see \citealp{mailath2006repeated},
for a textbook exposition); and (3) non-equilibrium level-k analysis
of games of conflict, where agents can use pre-play communication
to deceive naive opponents (see, e.g., \citealp{crawford2003lying}).

\paragraph{Structure.}

The rest of the paper is organised as follows. Section \ref{Sect Model}
presents the model. The results are presented in Section \ref{sec:Results}.
In Section \ref{sec:Extensions} we extend the model to deal with
partial observability (formally analysed in Appendix \ref{sec:partial-observability})
and type-interdependent preferences (formally analysed in Appendix
\ref{Sect Interdependent}). We conclude in Section \ref{Sect discussion}.
Appendix \ref{sec:Formal-Proofs-of} contains proofs not in the main
text. Appendix \ref{sec:Constructions-of-Heterogeneous} formally
constructs heterogeneous stable populations in specific games. 

\section{Model\label{Sect Model}}

We consider a large population of agents, each of whom is endowed
with a type that determines her subjective preferences and her cognitive
level. The agents are randomly matched to play a symmetric two-player
game. A dynamic evolutionary process of cultural learning, or biological
inheritance, increases the frequency of more successful types. We
present a static solution concept to capture stable population states
in such environments.

\subsection{Underlying Game and Types\label{subsec:Underlying-Game}}

Consider a symmetric two-player normal form game $G$ with a finite
set $A$ of pure actions and a set $\Delta\left(A\right)$ of mixed
actions (or strategies). We use the letter $a$ (resp., $\sigma$)
to describe a typical pure action (resp., mixed action). Payoffs are
given by $\pi:A\times A\rightarrow\mathbb{R}$, where $\pi\left(a,a^{\prime}\right)$
is the material (or fitness) payoff to a player using action $a$
against action $a^{\prime}$. The payoff function is extended to mixed
actions in the standard way, where $\pi\left(\sigma,\sigma^{\prime}\right)$
denotes the material payoff to a player using strategy $\sigma$,
against an opponent using strategy $\sigma^{\prime}$. With a slight
abuse of notation let $a$ denote the degenerate strategy that puts
all the weight on action $a$. We adopt this convention for probability
distributions throughout the paper. 
\begin{rem}
\label{rem:The-restriction-to}Asymmetric interactions can be captured
in our setup (as is standard in the literature; see, e.g., \citealp[Section 9.5]{brown1950solutions,selten1980note,van1987stability})
by embedding the asymmetric interaction in a larger, symmetric game
in which nature first randomly assigns the players to roles in the
asymmetric interaction. 
\end{rem}
We imagine a large population of individuals (technically, a continuum)
who are uniformly randomly matched to play the game $G$. Each individual
$i$ in the population is endowed with a \emph{type} $\theta=\left(u,n\right)\in\Theta=U\times\mathbb{N}\text{,}$
consisting of \emph{preferences}, identified with a von Neumann\textendash Morgenstern
utility function, $u\in U$, and \emph{cognitive level}\footnote{For tractability, we choose to work with a discrete set of cognitive
levels. The main results in the paper can be adapted to a setup in
which the feasible set of cognitive efforts is a continuum, provided
that we maintain our focus on finite-support type distributions. } $n\in\mathbb{N}$. Let $\Delta\left(\Theta\right)$ be the set of
all finite support probability distributions on $\Theta$. A population\textit{
}is represented by a finite-support \emph{type distribution} $\mu\in\Delta\left(\Theta\right)$.
\footnote{Comment \ref{enu:Continuum-population-and} in Section \ref{subsec:Configurations}
explains why we restrict attention to finite-support type distributions.} Let $C\left(\mu\right)$ denote the support (carrier) of type distribution
$\mu\in\Delta\left(\Theta\right)$. Given a type $\theta$, we use
$u_{\theta}$ and $n_{\theta}$ to refer to its preferences and cognitive
level, respectively.

In the main model we assume that the preferences are defined over
action profiles, as in \citet{Dekel_Ely_Yilankaya_2006}.\footnote{In Appendix \ref{Sect Interdependent}, we study \emph{type-interdependent}
preferences, which depend on the opponent's type, as in \citet{Herold_Kuzmics_2009}. } This means that any preferences can be represented by a utility function
of the form $u:A\times A\rightarrow\mathbb{R}\text{.}$ The set of
all possible (modulo affine transformations) utility functions on
$A\times A$ is $U=\left[0,1\right]^{\left\vert A\right\vert ^{2}}$.
Let $BR_{u}\left(\sigma^{\prime}\right)$ denote the set of best replies
to strategy $\sigma^{\prime}$ given preferences $u$, i.e. $BR_{u}\left(\sigma^{\prime}\right)=\arg\max_{\sigma\in\Delta\left(A\right)}u\left(\sigma,\sigma^{\prime}\right)$.

There is a fitness cost to increased cognition, represented by a strictly
increasing cognitive cost function $k:\mathbb{N\rightarrow R}_{+}$
satisfying $\lim_{n\rightarrow\infty}k\left(n\right)=\infty$. The
fitness payoff of an individual equals the material payoff from the
game, minus the cognitive cost. Let $k_{n}$ denote the cost of having
cognitive level $n$. Hence $k_{\theta}=k_{n_{\theta}}$ denotes the
cost of having type $\theta$. Without loss of generality, we assume
that $k_{1}=0$. 

We would like to put forward two motivations for the assumption that
there is an increasing fitness cost of having a higher cognitive level.
The first motivation is relevant to settings in which the evolution
of types is influenced by biological inheritance. There is a literature
in biology and biological anthropology showing that brain volume,
especially neocortex volume, is correlated with the size of social
groups across species. Noting that brain tissue is metabolically costly,
it has been argued that the size of the brain (in particular the neocortex)
is at least partially determined by complexity of social organisation
(see \citealp{Dunbar_1998}, for a summary of the evidence and the
arguments), which is in line with the ``Machiavellian intelligence\textquotedblright{}
and ``social brain\textquotedblright{} hypotheses (\citealp{Humphrey_1976,byrne1997machiavellian,whiten_byrne_1988tactical}).

The second motivation is relevant also in setups in which types evolve
as part of a social learning process. For concreteness, suppose that
agents face two kinds of decision problems throughout their lives:
(1) individual (ecological) decision problems against nature, and
(2) interactive (social) decision problems as represented by playing
the underlying game $G$. Agents have limited cognitive capacity.
New agents who join the population face a trade-off between developing
their deception-related cognitive skills (which are helpful when playing
the game $G$) and developing other skills (which are helpful in the
decision problems against nature). When a new agent joins the population,
his type $\theta=\left(u_{\theta},n_{\theta}\right)$ determines how
much effort the agent exerts in developing his deception-related cognitive
ability $n_{\theta}$ (while the remaining effort is exerted to develop
the other skills). The increasing cognitive cost function $k\left(n_{\theta}\right)$
captures the agent's loss due to his sub-optimal performance in the
decision problems against nature, which is induced by diverting effort
to developing his deception-related cognitive ability at the expense
of developing the other skills.

\subsection{Configurations\label{subsec:Configurations}}

A state of the population is described by a type distribution and
a behaviour policy for each type in the support of the type distribution.
An individual's behaviour is assumed to be (subjectively) rational
in the sense that it maximises her subjective preferences given the
belief she has about the opponent's expected behaviour. However, her
beliefs may be incorrect if she is deceived by her opponent. An individual
may be deceived if her opponent is of a strictly higher cognitive
level. The probability of deception is given by the function\emph{
$q:\mathbb{N}\times\mathbb{N}\rightarrow\left[0,1\right]$ }that satisfies
$q\left(n,n'\right)=0$ if and only if $n\leq n'$.\footnote{One can extend our main results to a setup in which individuals with
lower cognitive levels can deceive opponents with higher cognitive
levels with a sufficiently small probability. Specifically, assume
that for each generic game, there exists $\epsilon>0$, such that
$q\left(n,n'\right)<\epsilon$ for each $n\leq n'$ (instead of requiring
$q\left(n,n'\right)=0$). One can show that the characterization of
NSCs in Corollary \ref{cor:full-characterization-generic-games} remains
qualitatively the same. Namely, the only candidates to be NSCs are
configurations in which all agents have the minimal cognitive level,
and all agents play the efficient action profile in every match with
no deception. These configurations are NSCs if the effective cost
of defection is sufficiently low.} We interpret $q\left(n,n'\right)$ as the probability that a player
with cognitive level $n$ deceives an opponent with cognitive level
$n'$. Specifically, when two players with cognitive levels $n'$
and $n\geq n'$ are matched to play, then with a probability of $q\left(n,n'\right)$
the individual with the higher cognitive level $n$ (henceforth, the
\emph{higher type}) observes the opponent's preferences perfectly,
and is able to deceive the opponent (henceforth, the \emph{lower type}).
The deceiver is allowed to choose whatever she wants the deceived
party to believe about the deceiver's intended action choice. The
deceived party best-replies given her possibly incorrect belief. For
simplicity, we assume that if the deceived party has multiple best
replies, then the deceiver is allowed to break indifference, and choose
which of the best replies she wants the deceived party to play. Consequently
the deceiver is able to induce the deceived party to play any strategy
that is a best reply to some belief about the opponent's mixed action,
given the deceived party's preferences. 

Given preferences $u\in U$, let $\Sigma\left(u\right)$ denote the
set of \emph{undominated strategies}, which are the set of actions
that are best replies to at least one strategy of the opponent (given
the preferences $u$). Formally, we define 
\[
\Sigma\left(u\right)=\left\{ \sigma\in\Delta\left(A\right):\text{ there exists }\sigma^{\prime}\in\Delta\left(A\right)\text{ such that }\sigma\in BR_{u}\left(\sigma^{\prime}\right)\right\} .
\]
We say that a strategy profile is a \emph{deception equilibrium} if
the strategy profile is optimal from the point of view of the deceiver
under the constraint that the deceived player has to play an undominated
strategy. Formally:
\begin{defn}
Given two types $\theta,\theta^{\prime}$ with $n_{\theta}>n_{\theta^{\prime}},$
a strategy profile $\left(\tilde{\sigma},\tilde{\sigma}^{\prime}\right)$
is a \emph{deception equilibrium }if
\[
\left(\tilde{\sigma},\tilde{\sigma}^{\prime}\right)\in\arg\max_{\sigma\in\Delta\left(A\right),\sigma^{\prime}\in\Sigma\left(u_{\theta^{\prime}}\right)}u_{\theta}\left(\sigma,\sigma^{\prime}\right).
\]
Let $DE\left(\theta,\theta^{\prime}\right)$ be the set of all such
deception equilibria. 

With the remaining probability of $1-q\left(n,n'\right)-q\left(n',n\right)$
there is no deception between the players with cognitive levels $n$
and $n'$, and they play a Nash equilibrium of the game induced by
their preferences. Given two preferences $u,u^{\prime}\in U$, let
$NE\left(u,u^{\prime}\right)\subseteq\Delta\left(A\right)\times\Delta\left(A\right)$
be the set of mixed equilibria of the game induced by the preferences
$u$ and $u^{\prime}$, i.e.
\[
NE\left(u,u^{\prime}\right)=\left\{ \left(\sigma,\sigma^{\prime}\right)\in\Delta\left(A\right)\times\Delta\left(A\right):\sigma\in BR_{u}\left(\sigma^{\prime}\right)\text{ and }\sigma^{\prime}\in BR_{u^{\prime}}\left(\sigma\right)\right\} .
\]
\end{defn}
We are now in a position to define our key notion of a configuration
(following the terminology of \citealp{Dekel_Ely_Yilankaya_2006}),
by combining a type distribution with a behaviour policy, as represented
by Nash equilibria and deception equilibria.
\begin{defn}
A\emph{ configuration} is a pair $\left(\mu,b\right)$ where $\mu\in\Delta\left(\Theta\right)$
is a type distribution, and $b=\left(b^{N},b^{D}\right)$ is a \emph{behaviour
policy}, where $b^{N},b^{D}:C\left(\mu\right)\times C\left(\mu\right)\longrightarrow\Delta\left(A\right)$
satisfy for each $\theta,\theta^{\prime}\in C\left(\mu\right):$ 
\[
q\left(n_{\theta},n_{\theta'}\right)+q\left(n_{\theta'},n_{\theta}\right)<1\,\,\Rightarrow\,\,\left(b_{\theta}^{N}\left(\theta^{\prime}\right),b_{\theta^{\prime}}^{N}\left(\theta\right)\right)\in NE\left(\theta,\theta^{\prime}\right),\text{ and}
\]
\[
{\normalcolor {\color{purple}{\normalcolor q\left(n_{\theta},n_{\theta'}\right)}{\normalcolor >0}{\normalcolor \,\Leftrightarrow}}}\,\ensuremath{n_{\theta}>n_{\theta'}}\,\Rightarrow\,\left(b_{\theta}^{D}\left(\theta^{\prime}\right),b_{\theta^{\prime}}^{D}\left(\theta\right)\right)\in DE\left(\theta,\theta^{\prime}\right).
\]
We interpret $b_{\theta}^{D}\left(\theta^{\prime}\right)$ (resp.,
$b_{\theta}^{N}\left(\theta^{\prime}\right)$) to be the strategy
used by type $\theta$ against type $\theta^{\prime}$ when deception
occurs (resp., does not occur).
\end{defn}
Given a configuration $\left(\mu,b\right)$ we call the types in the
support of $\mu$ \emph{incumbents}. Note that standard arguments
imply that for any distribution $\mu$, there exists a mapping $b:C\left(\mu\right)\times C\left(\mu\right)\rightarrow\Delta\left(A\right)$
such that $\left(\mu,b\right)$ is a configuration. Given a configuration
$\left(\mu,b\right)$ and types $\theta,\theta'\in C\left(\mu\right)$,
let $\pi_{\theta}\left(\theta'|\left(\mu,b\right)\right)$ be the
expected fitness of an agent with type $\theta$ conditional on being
matched with $\theta'$: {\small{}
\[
\pi_{\theta}\left(\theta'|\left(\mu,b\right)\right)=\left(q\left(n_{\theta},n_{\theta'}\right)+q\left(n_{\theta'},n_{\theta}\right)\right)\cdot\pi\left(b_{\theta}^{D}\left(\theta^{\prime}\right),b_{\theta^{\prime}}^{D}\left(\theta\right)\right)+\left(1-\left(q\left(n_{\theta},n_{\theta'}\right)+q\left(n_{\theta'},n_{\theta}\right)\right)\right)\cdot\pi\left(b_{\theta}^{N}\left(\theta^{\prime}\right),b_{\theta^{\prime}}^{N}\left(\theta\right)\right).
\]
} 

The expected fitness of an individual of type $\theta$ in configuration
$\left(\mu,b\right)$ is
\[
\Pi_{\theta|\left(\mu,b\right)}=\sum_{\theta^{\prime}\in C\left(\mu\right)}\mu\left(\theta'\right)\cdot\pi_{\theta}\left(\theta'|\left(\mu,b\right)\right)-k_{\theta},
\]
where $\mu\left(\theta'\right)$ denotes the frequency of type $\theta'$
in the population. Given a configuration $\left(\mu,b\right)$, let
$\Pi_{\left(\mu,b\right)}$ be the average fitness in the population,
i.e., 
\[
\Pi_{\left(\mu,b\right)}=\sum_{\theta\in C\left(\mu\right)}\mu\left(\theta\right)\cdot\Pi_{\theta|\left(\mu,b\right)}.
\]
When all incumbent types have the same expected fitness (i.e. $\Pi_{\left(\mu,b\right)}=\Pi_{\theta|\left(\mu,b\right)}$
for each $\theta\in C\left(\mu\right)$), we say that the configuration
is \emph{balanced}.

A number of aspects of our model of cognitive sophistication merit
further discussion.
\begin{enumerate}
\item \emph{Unidimensional cognitive ability}: In reality the ability to
deceive and the ability to detect preferences are probably not identical.
However, both of them are likely to be strongly related to cognitive
ability in general, and more specifically to theory of mind and the
ability to entertain higher-order intentional attitudes (\citealp{Kinderman_etal_1998};
\citealp{Dunbar_1998}). For this reason we believe that a unidimensional
cognitive trait is a reasonable approximation. Moreover, it is an
approximation that affords us necessary tractability. We connect the
abilities to detect and conceal preferences with the ability to deceive,
by assuming (throughout the paper) that one is able to deceive one\textquoteright s
opponent if and only if one observes the opponent\textquoteright s
preferences and conceals one\textquoteright s own preferences from
the opponent. 
\item \emph{\label{enu:Power-of-deception:}Power of deception}: Our definition
of deception equilibrium amounts to an assumption that a successful
deception attempt allows the deceiver to implement her favourite strategy
profile, under the constraint that the deceived party does not choose
a dominated action from her point of view. Moreover, we assume that
a player with a higher cognitive level knows whether her deception
was successful when choosing her action. These assumptions give higher
cognitive types a clear advantage over lower cognitive types. Hence,
in an alternative model in which successful deceivers have less deception
power, we would expect the evolutionary advantage of higher types
to be weaker than in our current model. Below we find that (for generic
games) in any stable state everyone plays the same efficient action
profile and has the lowest cognitive level.\footnote{Thus, in our setup a cognitive arms race (i.e. Machiavellian intelligence
hypothesis $\grave{\textrm{a}}$ la \citealp{Humphrey_1976,Robson_2003})
is a non-equilibrium phenomenon, or alternatively a feature of non-generic
games.} We conjecture that these states will remain stable also in a model
where successful deception is less powerful. We leave for future research
the analysis of feasible but less powerful deception technologies.
\item \emph{Same deception against all lower types}: Our model assumes that
a player may use different deceptions against different types with
lower cognitive levels. We note that our results remain the same (with
minor changes to the proofs) in an alternative setup in which individuals
have to use the \emph{same} mixed action in their deception efforts
towards all opponents. 
\item \emph{Non-Bayesian deception}: Note that a successful deceiver is
able to induce the opponent to believe that the deceiving type will
play any mixed action $\hat{\sigma}$, even an action that is never
played by any agent in the population. That is, deception is so powerful
in our model that the deceived opponent is not able to apply Bayesian
reasoning in his false assessment of which action the agent is going
to play. We think of this assumption as describing a setting in which
the deceiver (of a higher cognitive type) is able to provide a convincing
argument (tell a convincing story) that she is going to play $\hat{\sigma}$.
From a Bayesian perspective one might object that these arguments
are signals that should be used to update beliefs. To this we would
respond that the stories told to a potential victim by different deceivers
will vary across would-be deceivers, even across would-be deceivers
with the same preferences. Hence no individual will ever accumulate
a database containing more than one or a handful of similar arguments.
The limited amount of data on similar arguments will preclude the
efficient use of Bayesian updating for inferring likely behaviour
following different arguments. We are not aware of the existence of
a Bayesian model of deception that is satisfactory for our purposes.
We leave the development of such a Bayesian model to future research.
\item \emph{Observation and Nash equilibrium behaviour in the case of non-deception}:
It is difficult to avoid an element of arbitrariness when making an
assumption about what is being observed when neither party is able
to deceive the other. As in most of the existing literature on the
indirect evolutionary approach (e.g., \citealp{Guth_Yaari1992}; \citealp[Section 3]{Dekel_Ely_Yilankaya_2006}),
we assume that when there is no deception, then there is perfect observability
of the opponent's preferences. In Section \ref{subsec:Partial-Observability-When}
we discuss the implications of the relaxation of this assumption.
We consider it to be an important contribution of our analysis that
it highlights the critical importance of the assumption made regarding
observability, and the resulting behaviour, in matches without deception.\\
We further assume that if two agents observe each other's preferences
then they play a Nash equilibrium of the complete information game
induced by their preferences. This assumption is founded on the common
idea that when agents are not deceived, then (1) over time they adapt
their beliefs (in a way that is consistent with Bayesian inference)
about the distribution of actions they face, conditional on their
partners' observed preferences, and (2) they best-reply given their
belief about their current partner's distribution of actions. By contrast,
as discussed above, when agents are deceived they are unable to correctly
update their beliefs about their partner's action (i.e. unable to
use Bayesian inference to arrive at beliefs about the opponent's distribution
of actions). Still, they are able to best-reply given their (possibly
false) beliefs about the deceiver's action. 
\item \emph{\label{enu:Continuum-population-and}Continuum population and
finite-support type distributions}. Our model is intended to be a
simple approximation of a real-life environment that includes a large
finite population, and in which new agents who join the population,
or existing agents who revise their choice of type, typically choose
to mimic one of the existing active types. As a result each active
type is played by several agents (rather than by a single agent),
and for each active type there is a positive probability of a match
between agents who are endowed with this type. As is common in the
literature, for tractability, we assume a continuum population and
an ``exact law of large numbers,'' rather then a large finite population.
We want all other aspects of the model to be as close as possible
to the real-life environment. Specifically, we want to maintain the
property that for each type, there is a positive probability of a
match between agents who are endowed with this type. In order to maintain
this property, we have to assume that the distribution of active types
has a finite support.\footnote{More accurately, we need to assume that the set of active types is
countable. All of our results hold under this somewhat weaker assumption. } 
\item \label{enu:An-alternative-interpretation}\emph{Alternative interpretation
of our model: social status}. As suggested by one of the referees,
one can present an interesting interpretation of our model that describes
social status, rather than deception. According to this interpretation,
the level $n_{\theta}$ of type $\theta$ describes the social status
(like caste) of agents belonging to this type. When two players are
randomly matched to play a game, first a \textquotedblleft social
struggle\textquotedblright{} ensues. With a certain probability, the
higher-caste player prevails and enslaves the lower-caste opponent.
This means he can dictate the choice by the lower-caste opponent as
long as the choice is undominated for this opponent. Otherwise, they
simply play the Nash equilibrium of the game (given by their preferences).
Maintaining a higher social status is costly in terms of fitness. 
\end{enumerate}

\subsection{Evolutionary Stability \label{subSect evolution}}

As discussed in the previous subsection, each agent in the population
behaves in a way that maximises the agent's subjective preferences
induced by the agent's type. By contrast, the distribution of types
in the population evolves according to the expected material fitness
obtained by each type. This evolutionary process is captured by the
static solution concepts introduced in this subsection.

We consider  dynamics in which types with higher expected fitness
gradually become more frequent. One example of such dynamics is the
replicator dynamic (\citealp{Taylor_Jonker_1978}), which can be interpreted
in terms of biological (asexual) reproduction or as social learning
by imitation (see \citealp[Chapter 3]{Weibull_1995} for a textbook
introduction). According to the latter interpretation, an agent who
has the opportunity to revise her choice or a new agent who joins
the population randomly chooses a member of the population as ``mentor,''
and imitates the mentor's type; the probability that an agent is chosen
as a mentor is proportional to that agent's fitness.

Recall that a neutrally stable strategy (\citealp{Maynard_Smith_Price_1973};
\citealp{Maynard_smith1982evolution}) is a strategy that, if played
by most of the population, weakly outperforms any other strategy.
Similarly, an evolutionarily stable strategy is a strategy that, if
played by most of the population, strictly outperforms any other strategy. 
\begin{defn}
A strategy $\sigma\in\Delta\left(A\right)$ is a \emph{neutrally stable
strategy (NSS)}\textbf{ }if for every $\sigma^{\prime}\in\Delta\left(A\right)$
there is some $\bar{\varepsilon}\in\left(0,1\right)$ such that if
$\varepsilon\in\left(0,\bar{\varepsilon}\right)$, then $\tilde{\pi}\left(\sigma^{\prime},\left(1-\varepsilon\right)\sigma+\varepsilon\sigma^{\prime}\right)\leq\tilde{\pi}\left(\sigma,\left(1-\varepsilon\right)\sigma+\varepsilon\sigma^{\prime}\right)$.
If weak inequality is replaced by strict inequality for each $\sigma^{\prime}\neq\sigma,$
then $\sigma$ is an \emph{evolutionarily stable strategy (ESS)}. 
\end{defn}
It is well known that NSSs and ESSs correspond to Lyapunov stable
and asymptotically stable population states, respectively, under the
replicator dynamics. That is, a population starting close to an NSS
will always remain close to the NSS, and a population starting close
to an ESS will converge to the ESS (see, e.g., \citealp{Taylor_Jonker_1978,Thomas_1985_ESset,bomze1995does,Cressman_1997,Sandholm_2010_localstability}.) 

We extend the notions of neutral and evolutionary stability from strategies
to configurations. We begin by defining the type game that is induced
by a configuration.
\begin{defn}
\label{Def type game}For any configuration $\left(\mu,b\right)$
the corresponding \emph{type game} $\Gamma_{\left(C\left(\mu\right),b\right)}$
is the symmetric two-player game where each player's pure strategy
space is $C\left(\mu\right)$, and the payoff to strategy $\theta$,
against $\theta^{\prime}$, is $\pi_{\theta}\left(\theta'|\left(\mu,b\right)\right)-k_{\theta}$. 
\end{defn}
The definition of a type game allows us to apply notions and results
from standard evolutionary game theory, where evolution acts upon
strategies, to the present setting where evolution acts upon types.
A similar methodology was used in \citet{Mohlin_2012_ETM}. Note that
each type distribution with support in $C\left(\mu\right)$ is represented
by a mixed strategy in $\Gamma_{\left(C\left(\mu\right),b\right)}$.

We want to capture robustness with respect to small groups of individuals,
henceforth called \textit{mutants}, who introduce new types and new
behaviours into the population. Suppose that a fraction $\varepsilon$
of the population is replaced by mutants and suppose that the distribution
of types within the group of mutants is $\mu^{\prime}\in\Delta\left(\Theta\right)$.
Consequently the post-entry type distribution is $\tilde{\mu}=\left(1-\varepsilon\right)\cdot\mu+\varepsilon\cdot\mu^{\prime}$.
That is, for each type $\theta\in C\left(\mu\right)\cup C\left(\mu^{\prime}\right)$,
$\tilde{\mu}\left(\theta\right)=\left(1-\varepsilon\right)\cdot\mu\left(\theta\right)+\varepsilon\cdot\mu^{\prime}\left(\theta\right)$.
In line with most of the literature on the indirect evolutionary approach
we assume that adjustment of behaviour is infinitely faster than adjustment
of type distribution.\footnote{\citet{Sandholm_2001} and \citet{Mohlin_2010_EoG} are exceptions.}
Thus we assume that the post-entry type distribution quickly stabilises
into a configuration $\left(\tilde{\mu},\tilde{b}\right)$. There
may exist many such post-entry type configurations, all having the
same type distribution, but different behaviour policies. We note
that incumbents do not have to adjust their behaviour against other
incumbents in order to continue playing Nash equilibria, and deception
equilibria, among themselves. For this reason, we assume (similarly
to \citealp{Dekel_Ely_Yilankaya_2006}, in the setup with perfect
observability) that the incumbents maintain the same pre-entry behaviour
among themselves. Formally:
\begin{defn}
Let $\left(\mu,b\right)$ and $\left(\tilde{\mu},\tilde{b}\right)$
be two configurations such that $C\left(\mu\right)\subseteq C\left(\tilde{\mu}\right)$.
We say that $\left(\tilde{\mu},\tilde{b}\right)$ is \emph{focal}
(with respect to $\left(\mu,b\right))$ if $\theta,\theta^{\prime}\in C\left(\mu\right)$
implies that $\tilde{b}_{\theta}^{D}\left(\theta^{\prime}\right)=b_{\theta}^{D}\left(\theta^{\prime}\right)$
and $\tilde{b}_{\theta}^{N}\left(\theta^{\prime}\right)=b_{\theta}^{N}\left(\theta^{\prime}\right)$.
\end{defn}
Standard fixed-point arguments imply that for every configuration
$\left(\mu,b\right)$ and every type distribution $\tilde{\mu}$ satisfying
$C\left(\mu\right)\subseteq C\left(\tilde{\mu}\right),$ there exists
a behaviour policy $\tilde{b}$ such that $\left(\tilde{\mu},\tilde{b}\right)$
is a focal configuration.

Our stability notion requires that the incumbents outperform all mutants
in all configurations that are focal relative to the initial configuration.
\begin{defn}
\label{Def ESC}A configuration $\left(\mu,b\right)$ is a \emph{neutrally
stable configuration (NSC)} if, for every $\mu^{\prime}\in\Delta\left(\Theta\right)$,
there is some $\bar{\varepsilon}\in\left(0,1\right)$ such that for
all $\varepsilon\in\left(0,\bar{\varepsilon}\right)$, it holds that
if $\left(\tilde{\mu},\tilde{b}\right)$, where $\tilde{\mu}=\left(1-\varepsilon\right)\cdot\mu+\varepsilon\cdot\mu^{\prime}$,
is a focal configuration, then $\mu$ is an NSS in the type game $\Gamma_{\left(\tilde{\mu},\tilde{b}\right)}$.
The configuration $\left(\mu,b\right)$ is an \emph{evolutionarily
stable configuration (ESC)}\textbf{ }if the same conditions imply
that $\mu$ is an ESS in the type game $\Gamma_{\left(\tilde{\mu},\tilde{b}\right)}$
for each $\mu^{\prime}\neq\mu$. 
\end{defn}
We conclude this section by discussing a few issues related to our
notion of stability.
\begin{enumerate}
\item In line with existing notions of evolutionary stability in the literature
(in particular, the notions of \citealp{Dekel_Ely_Yilankaya_2006},
and \citealp{Alger_Weibull_HomoMoralis}), we require the mutants
to be outperformed in all focal configurations (rather than requiring
them to be outperformed in at least one focal configuration). This
reflects the assumption that the population converges to a new post-entry
equilibrium in a decentralised (possibly random) way that may lead
to any of the post-entry focal configurations. Thus the incumbents
cannot coordinate their post-entry play on a specific focal configuration
that favors them.
\item In order to be consistent with the standard definition of neutral
stability, we require the incumbents to earn weakly more than the
average payoff of the mutants. We note that all of our results remain
the same if one uses an alternative weaker definition that requires
the incumbents to earn weakly more than the worst-performing mutant.
\item The main stability notion that we use in the paper is NSC. The stronger
notion of ESC is not useful in our main model because there always
exist equivalent types that have slightly different preferences (as
the set of preferences is a continuum) and induce the same behaviour
as the incumbents. Such mutants always achieve the same fitness as
the incumbents in post-entry configurations, and thus ESCs never exist.
Note that the stability notions in \citet{Dekel_Ely_Yilankaya_2006}
and \citet{Alger_Weibull_HomoMoralis} are also based on neutral stability.\footnote{In their stability analysis of \textit{homo hamiltonensis} preferences
\citet{Alger_Weibull_HomoMoralis} disregard mutants who are behaviourally
indistinguishable from \textit{homo hamiltonensis} upon entry.} In Section \ref{Sect Interdependent} we study a variant of the model
in which the preferences may depend also on the opponent's types.
This allows for the existence of ESCs.
\item \label{Remark-internal-stability}Observe that Definition \ref{Def ESC}
implies internal stability with respect to small perturbations in
the frequencies of the incumbent types (because when $\mu^{\prime}=\mu,$
then $\mu$ is required to be an NSS in $\Gamma_{\left(C\left(\mu\right),b\right)}$).
By standard arguments, internal stability implies that any NSC is
balanced: all incumbent types obtain the same fitness.
\item \label{Remark-monomorphic-mutants}The stability notions of \citet{Dekel_Ely_Yilankaya_2006}
and \citet{Alger_Weibull_HomoMoralis} consider only monomorphic groups
of mutants (i.e. mutants all having the same type). We additionally
consider stability against polymorphic groups of mutants (as do \citealp{Herold_Kuzmics_2009}).
One advantage of our approach is that it allows us to use an adaptation
of the well-known notion of ESS, which immediately implies dynamic
stability and internal stability, whereas \citet{Dekel_Ely_Yilankaya_2006}
have to introduce a novel notion of stability without these properties.
Remark \ref{Rem-monomorphic-stability} below discusses the influence
on our results of using an alternative definition that deals only
with monomorphic mutants. 
\end{enumerate}

\section{Results\label{sec:Results} }

\subsection{Preliminary Definitions}

Define the \emph{deviation gain} of action $a\in A$, denoted by $g\left(a\right)\in\mathbb{R}^{+}$,
as the maximal gain a player can get by playing a different action
in a population in which everyone plays $a$: 
\[
g\left(a\right)=\max_{a'\in A}\pi\left(a',a\right)-\pi\left(a,a\right).
\]
Note that $g\left(a\right)=0$ iff $\left(a,a\right)$ is a Nash equilibrium.

Define the \emph{effective cost of deception }in the environment,
denoted by $c\in\mathbb{R}^{+}$, as the minimal ratio between the
cognitive cost and the probability of deceiving an opponent of cognitive
level one:\footnote{The minimum in the definition of $c$ is well defined for the following
reason. Let $\hat{n}$ be a number such that $k_{\hat{n}}>\frac{k_{2}}{q\left(2,1\right)}$
(such a number exists because $\lim_{n\rightarrow\infty}k_{n}=\infty$).
Observe that $\frac{k_{n}}{q\left(n,1\right)}\geq k_{n}>\frac{k_{2}}{q\left(2,1\right)}$
for any $n\geq\hat{n}$. This implies that there is an $\bar{n}$
such that $2\leq\bar{n}\leq\hat{n}$ and $\bar{n}=\textrm{arg}\min_{n\geq2}\,\,\frac{k_{n}}{q\left(n,1\right)}$.} \footnote{We define the effective cost of deception only with respect to an
opponent with a cognitive level of one because we later show (Lemma
\ref{Lemma: level 1-1} and Theorem \ref{thm-main-result-efficiency})
that the only candidate to be an NSC is a configuration in which all
agents have a cognitive level of one, and such a configuration is
an NSC iff the effective cost of defection against these incumbents
with $n=1$ is sufficiently large.}
\[
c=\min_{n\geq2}\,\,\frac{k_{n}}{q\left(n,1\right)}.
\]

We say that a strategy profile is efficient if it maximises the sum
of fitness payoffs. Formally:
\begin{defn}
\label{Def: efficeincy}A strategy profile $\left(\sigma,\sigma^{\prime}\right)$
is \emph{efficient }\textit{\emph{in the game }}$G=\left(A,\pi\right)$
if $\pi\left(\sigma,\sigma^{\prime}\right)+\pi\left(\sigma^{\prime},\sigma\right)\geq\pi\left(a,a^{\prime}\right)+\pi\left(a^{\prime},a\right)$,
for each action profile $\left(a,a^{\prime}\right)\in A^{2}$. 
\end{defn}
Note that our notion of efficiency is defined: (1) with respect to
the fitness payoff (rather than the agents' subjective payoffs), similarly
to the analogous definition of efficiency in \citet{Dekel_Ely_Yilankaya_2006},
and (2) with respect to the strategy profile played by the agents;
by contrast, the definition does not take into account the cognitive
costs. 

A pure Nash equilibrium $\left(a,a\right)$ is\emph{ strict }if $\pi\left(a,a\right)>\pi\left(a^{\prime},a\right)$
for all $a^{\prime}\neq a\in A$. Let $\hat{\pi}=\max_{a,a'\in A}\left(0.5\cdot\left(\pi\left(a,a^{\prime}\right)+\pi\left(a^{\prime},a\right)\right)\right)$
denote the efficient payoff, i.e. the average payoff achieved by players
who play an efficient profile. 

An action $a$ is a\emph{ punishment action} if\emph{ }playing it
guarantees that the opponent will obtain less than the efficient payoff,
i.e. $\pi\left(a',a\right)<\hat{\pi}$ for each $a'\in A$. Some of
our results below assume that the underlying game admits a punishment
action.
\begin{rem}
\label{rem:punishment-actions}Many economic interactions admit punishment
actions. A few examples include:
\begin{enumerate}
\item Price competition (Bertrand), either for a homogeneous good or for
differentiated goods, where a punishment action is setting the price
equal to zero.
\item Quantity competition (Cournot), either for a homogeneous good or for
differentiated goods, where the punishment action is ``flooding''
the market. 
\item Public good games, where contributing nothing to the public good is
the punishment action.
\item Bargaining situations, where the punishment action is for one side
of the bargaining to insist on obtaining all surplus.
\item Any game that admits an action profile that Pareto dominates all other
action profiles (i.e., games with common interests). 
\end{enumerate}
Moreover, if one adds to any underlying generic game a new pure action
that is equivalent to playing the mixed action that min-maxes the
opponent's payoff (e.g., in matching pennies this new action is equivalent
to privately tossing a coin and then playing according to the toss's
outcome), then this newly added action is always a punishment action.
\end{rem}
Given a configuration $\left(\mu,b\right)$ let $\bar{n}=\max_{\theta\in C\left(\mu\right)}n_{\theta}$
denote the maximal cognitive level of the incumbents. We refer to
incumbents with this cognitive level as the \emph{highest types}.

A deception equilibrium is \emph{fitness maximising} if it maximises
the fitness of the higher type in the match (under the restriction
that the lower type plays an action that is not dominated, given her
preferences). Formally:
\begin{defn}
\label{def:FMDE}Let $\theta,\theta'$ be types with $n_{\theta}>n_{\theta'}$.
A deception equilibrium $\left(\tilde{\sigma},\tilde{\sigma}'\right)$
is \emph{fitness maximising} if
\[
\left(\tilde{\sigma},\tilde{\sigma}'\right)\in\arg\max_{\sigma\in\Delta\left(A\right),\,\sigma'\in\Sigma\left(u_{\theta'}\right)}\pi\left(\sigma,\sigma'\right).
\]

Let $FMDE\left(\theta,\theta'\right)\subseteq DE\left(\theta,\theta'\right)$
denote the set of all such fitness-maximising deception equilibria
of two types $\theta,\theta'$ with $n_{\theta}>n_{\theta'}$. In
principle, $FMDE\left(\theta,\theta'\right)$ might be an empty set
(if there is no action profile that maximises both the fitness and
the subjective utility of the higher type). Our first result (Theorem
\ref{thm:highest-type-behaivor} below) implies that the preference
of the higher type in any NSC are such that the set $FMDE\left(\theta,\theta'\right)$
is non-empty. 
\end{defn}
A configuration is pure if everyone plays the same action. Formally:
\begin{defn}
\label{def:pure-configuration}A configuration $\left(\mu,b\right)$
is \emph{pure} if there exists $a^{\ast}\in A$ such that for each
$\theta,\theta^{\prime}\in C\left(\mu\right)$ it holds that $b_{\theta}^{N}\left(\theta^{\prime}\right)=a^{\ast}$
whenever $q\left(\theta,\theta'\right)<1$, and $b_{\theta}^{D}\left(\theta^{\prime}\right)=a^{\ast}$
whenever $q\left(\theta,\theta'\right)>0$ . With a slight abuse of
notation we denote such a pure configuration by $\left(\mu,a^{\ast}\right)$,
and we refer to $b\equiv a^{\ast}$ as the \emph{outcome} of the configuration.
\end{defn}
In order to simplify the notation and the arguments in the proofs,
we assume throughout this section that the underlying game admits
at least three actions (i.e. $\left|A\right|\geq3)$. If the original
game has only two actions, then adding a third action, which is dominated
by the other two actions, allows all the arguments in the proof to
work. More complicated (and less instructive) variants of the proofs
can also be applied to a game with two actions without adding a third,
dominated action.

\subsection{Characterisation of the Highest Types' Behaviour\label{subsec:Characterization-of-Highest}}

In this section we characterise the behaviour of an incumbent type,
$\bar{\theta}=\left(u,\bar{n}\right)$, which has the highest level
of cognition in the population.\footnote{For tractability we assume that a configuration can have only finite
support. Note, however, that there is some sufficiently high cognitive
level $n$ such that $k_{n}>\max_{a,a'\in A}\pi\left(a,a'\right)$.
As a result, even if one relaxes the assumption of finite support,
any NSC must include only a finite number of cognitive levels, also
without the finite-support assumption.} We show that the behaviour satisfies the following three conditions: 
\begin{enumerate}
\item Type $\bar{\theta}$ plays an efficient action profile when meeting
itself.
\item Type $\bar{\theta}$ maximises its fitness in all deception equilibria.
\item Any opponent with a lower cognitive level achieves at most the efficient
payoff $\hat{\pi}$ against type $\bar{\theta}$.
\end{enumerate}
\begin{thm}
\label{thm:highest-type-behaivor}Let $\left(\mu^{*},b^{*}\right)$
be an NSC, and let $\underline{\theta},\bar{\theta}\in C\left(\mu^{*}\right)$.
Then: (1) if $n_{\bar{\theta}}=\bar{n}$ then $\pi\left(\bar{\theta},\bar{\theta}\right)=\hat{\pi}$;
(2) if $n_{\underline{\theta}}<n_{\bar{\theta}}=\bar{n}$ then $\left(b_{\bar{\theta}}^{D}\left(\underline{\theta}\right),b_{\underline{\theta}}^{D}\left(\bar{\theta}\right)\right)\in FMDE\left(\bar{\theta},\underline{\theta}\right)$;
and (3) if $n_{\underline{\theta}}<n_{\bar{\theta}}=\bar{n}$ then
$\pi\left(\underline{\theta},\bar{\theta}\right)\leq\hat{\pi}$.
\end{thm}
\begin{proof}[Proof Sketch (formal proof in Appendix \ref{subsec:Proof-of-Theorem-1})]
 The proof utilises mutants (denoted by $\theta_{1},\theta_{2},\theta_{3}$,
and $\hat{\theta}$, below) with the highest cognitive level $\bar{n}$
and with a specific kind of utility function, called \textit{indifferent
and pro-generous}, that makes a player indifferent between all her
own actions, but which makes the player prefer that the opponent choose
an action that allows the player to obtain the highest possible fitness
payoff. 

To prove part 1 of the theorem, assume to the contrary that $\pi\left(b_{\bar{\theta}}\left(\bar{\theta}\right),b_{\bar{\theta}}\left(\bar{\theta}\right)\right)<\hat{\pi}$.
Let $a_{1},a_{2}\in A$ be any two actions such that $\left(a_{1},a_{2}\right)$
is an efficient action profile (i.e. $0.5\cdot\left(\pi\left(a_{1},a_{2}\right)+\pi\left(a_{1},a_{2}\right)\right)=\hat{\pi}$).
Consider three different mutant types $\theta_{1}$, $\theta_{2}$,
and $\theta_{3}$, which are of the highest cognitive level that is
present in the population, and have indifferent and pro-generous utility
functions. Suppose equal fractions of these three mutant types enter
the population.\footnote{One must have at least two different types of mutants, in order for
the mutants to be able to play the asymmetric profile $\left(a_{1},a_{2}\right)$.
We preset a construction with three different mutant types in order
to allow all mutant types to outperform the incumbents (one can also
prove the result using a constructions with only two different mutant
types, but in this case one can only guarantee that the mutants, on
average, would outperform the incumbents) } There is a focal post-entry configuration in which the incumbents
keep playing their pre-entry play among themselves, the mutants play
the same Nash equilibria as the incumbent $\bar{\theta}$ against
all incumbent types (and the incumbents behave against the mutants
in the same way they behave against $\bar{\theta}$), the mutants
play fitness-maximising deception equilibria against all lower types,
when mutants of type $\theta_{i}$ are matched with mutants of type\footnote{If $i=1$ (resp., $i=2$, $i=3$), then $\theta_{\left(i+1\right)\,\textrm{mod}\,3}=\theta_{2}$
(resp., $\theta_{\left(i+1\right)\,\textrm{mod}\,3}=\theta_{3}$,
$\theta_{\left(i+1\right)\,\textrm{mod}\,3}=\theta_{1}$).} $\theta_{\left(i+1\right)\,\textrm{mod}\,3}$ they play the efficient
profile $\left(a_{1},a_{2}\right)$, and when two mutants of the same
type are matched they play the same way as two incumbents of type
$\bar{\theta}$ that are matched together. In such a focal post-entry
configuration all mutants earn a weakly higher fitness than $\bar{\theta}$
against the incumbents, and a strictly higher fitness against the
mutants. This implies that $\left(\mu^{*},b^{*}\right)$ cannot be
an NSC.

To prove part 2, assume to the contrary that $\left(b_{\bar{\theta}}^{D}\left(\underline{\theta}\right),b_{\underline{\theta}}^{D}\left(\bar{\theta}\right)\right)\not\in FMDE\left(\bar{\theta},\underline{\theta}\right)$.
Suppose mutants of type $\hat{\theta}$ enter. Consider a post-entry
configuration in which the incumbents keep playing their pre-entry
play among themselves, and the mutants mimic the play of $\bar{\theta}$,
except that they play fitness-maximising deception equilibria against
all lower types. The mutants obtain a weakly higher payoff than $\bar{\theta}$
against all types, and a strictly higher payoff than $\bar{\theta}$
against at least one lower type. Thus $\left(\mu^{*},b^{*}\right)$
cannot be an NSC.

To prove part 3, assume to the contrary that $\pi\left(\underline{\theta},\bar{\theta}\right)>\hat{\pi}$.
This implies that against type $\bar{\theta}$, type $\underline{\theta}$
earns more than $\hat{\pi}$ in either the deception equilibrium or
the Nash equilibrium. Suppose mutants of type $\hat{\theta}$ enter.
Consider a post-entry configuration in which the incumbents keep playing
their pre-entry play among themselves, while the mutants: (i) play
fitness-maximising deception equilibria against lower types, (ii)
mimic type $\underline{\theta}$'s play in the Nash/deception equilibrium
against type $\bar{\theta}$ in which $\underline{\theta}$ earns
more than $\hat{\pi}$, and (iii) mimic the play of $\bar{\theta}$
in all other interactions. The type $\hat{\theta}$ mutants earn strictly
more than $\bar{\theta}$ against both $\hat{\theta}$ and $\bar{\theta}$.
The mutants earn weakly more than $\bar{\theta}$ against all other
types. This implies that $\left(\mu^{*},b^{*}\right)$ cannot be an
NSC.
\end{proof}
\begin{rem}
\label{Rem-monomorphic-stability} The first part of Theorem \ref{thm:highest-type-behaivor}
(a highest type must play an efficient strategy profile when meeting
itself) is similar to \citeauthor{Dekel_Ely_Yilankaya_2006}'s \citeyearpar{Dekel_Ely_Yilankaya_2006}
Proposition 2, which shows that only efficient outcomes can be stable
in a setup with perfect observability and no deception. We should
note that \citet{Dekel_Ely_Yilankaya_2006} use a weaker notion of
efficiency. An action is efficient in the sense of \citet{Dekel_Ely_Yilankaya_2006}
(DEY-efficient) if its fitness is the highest among the symmetric
strategy profiles (i.e. action $a$ is DEY-efficient if $\pi\left(a,a\right)\geq\pi\left(\sigma,\sigma\right)$
for all strategies $\sigma\in\Delta\left(A\right)$). Observe that
our notion of efficiency (Definition \ref{Def: efficeincy}) implies
DEY-efficiency, but the converse is not necessarily true. The weaker
notion of DEY-efficiency is the relevant one in the setup of \citet{Dekel_Ely_Yilankaya_2006},
because they consider only monomorphic groups mutants; i.e. all mutants
who enter at the same time are of the same type. A similar result
would also hold in our setup if we imposed a similar limitation on
the set of feasible mutants. However, without such a limitation, heterogeneous
mutants can correlate their play, and our stronger notion of efficiency
is required to characterise stability. 
\end{rem}
An immediate corollary of Theorem \ref{thm:highest-type-behaivor}
is that a game that has only asymmetric efficient action profiles
does not admit any NSCs.
\begin{cor}
\label{cor:no-efficient-no-NSC}If $G$ does not have an efficient
profile that is symmetric (i.e. if $\pi\left(a,a\right)<\hat{\pi}$
for each $a\in A$), then the game does not admit an NSC.
\end{cor}
\begin{rem}
\label{rem:We-note-that}As discussed in Remark \ref{rem:The-restriction-to},
any interaction (symmetric or asymmetric) can be embedded in a larger,
symmetric game in which nature first randomly assigns roles to the
players, and then each player chooses an action given his assigned
role.\footnote{If the original game is symmetric, the role (i.e. being either the
row or the column player) can be interpreted as reflecting some observable
payoff-irrelevant asymmetry between the two players.} Observe that \emph{such an embedded game} \emph{always admits an
efficient symmetric action profile}. In particular, if the efficient
asymmetric profile in the original game is $\left(a,a'\right)$, then
the efficient symmetric profile in the embedded game is the one in
which each player plays $a$ as the row player and $a'$ as the column
player.
\end{rem}

\subsection{Characterisation of Pure NSCs\label{subsec:Characterization-of-Pure}}

In this subsection we characterise pure NSCs, i.e. stable configurations
in which everyone plays the same pure action in every match. Such
a configuration may be viewed as representing the state of a population
that has settled on a convention that there is a unique correct way
to behave. We begin by showing that in a pure NSC all incumbents have
the minimal cognitive level, since having a higher ability does not
yield any advantage when everyone plays the same action. 
\begin{lem}
\label{Lemma: level 1-1} If $\left(\mu,a^{\ast}\right)$ is an NSC,
and $\left(u,n\right)\in C\left(\mu\right)$, then $n=1$. 
\end{lem}
\begin{proof}
Since all players earn the same game payoff of $\pi\left(a^{\ast},a^{\ast}\right),$
they must also incur the same cognitive cost, or else the fitness
of the different incumbent types would not be balanced (which would
contradict that $\left(\mu,a^{\ast}\right)$ is an NSC). Moreover,
this uniform cognitive level must be level 1. Otherwise a mutant of
a lower level, who strictly prefers to play $a^{\ast}$ against all
actions, would strictly outperform the incumbents in nearby post-entry
focal configurations. 
\end{proof}
The following proposition shows that a pure outcome is stable iff
it is efficient and its deviation gain is smaller than than the effective
cost of deception. Formally:
\begin{prop}
\label{pro-full-characterization-pure-outcoes} Let $a^{*}$ be an
action in a game that admits a punishment action. The following two
statements are equivalent:

(a) There exists a type distribution $\mu$ such that $\left(\mu,a^{\ast}\right)$
is an NSC.

(b) $a^{*}$ satisfies the following two conditions: (1) $\pi\left(a^{*},a^{*}\right)=\hat{\pi}$,
and (2) $g\left(a^{*}\right)\leq c$. 
\end{prop}
\begin{proof}
~

\begin{enumerate}
\item \emph{``If side.''} Assume that $\left(a^{*},a^{*}\right)$ is an
efficient profile and that $g\left(a^{*}\right)\leq c$. Let $\widetilde{a}$
be a punishment action. Consider a monomorphic configuration $\left(\mu,a^{*}\right)$
consisting of type $\theta^{*}=\left(u^{\ast},1\right)$ where all
incumbents are of cognitive level 1 and of the same preference type
$u^{\ast}$, according to which all actions except $a^{*}$ and $\widetilde{a}$
are strictly dominated, $\widetilde{a}$ weakly dominates $a^{*}$,
and $a^{*}$ is a best reply to itself: 
\[
u^{*}\left(a,a'\right)=\begin{cases}
1 & \mbox{if}\,\,a=\widetilde{a}\,\textrm{\,\ and\,}\,a'\neq a^{*}\\
0 & \mbox{if}\,\,a=a^{*}\,\,\textrm{or}\,\,\left(a=\widetilde{a}\textrm{\,\,and}\,\,a'=a^{*}.\right)\\
-1 & \textrm{otherwise.}
\end{cases}
\]
Consider first mutants with cognitive level one. Observe that in any
post-entry configuration mutants with cognitive level one earns at
most $\hat{\pi}$ when they are matched with the incumbents, and strictly
less than $\hat{\pi}$ if the mutants play any action $a\neq a^{*}$
with positive probability against the incumbents. Further observe,
that the mutants can earn (on average) at most $\hat{\pi}$ when they
are matched with other mutants (because $\hat{\pi}$ is the efficient
payoff). This implies that incumbents weakly outperform any mutants
with cognitive level one in any post-entry population. \\
Next, consider mutants with a higher cognitive level $n>1$. Such
mutants can earn at most $\hat{\pi}+g\left(a^{*}\right)$ when they
deceive the incumbents and at most $\hat{\pi}$ when they do not deceive
the incumbents (recall that $\pi\left(\widetilde{a},\widetilde{a}\right)+g\left(\widetilde{a}\right)=max_{a'}\pi\left(a',\widetilde{a}\right)<\hat{\pi}$
because $\widetilde{a}$ is a punishment action). Thus the mutants
are weakly outperformed by the incumbents if
\[
q\left(n,1\right)\cdot\left(g\left(a^{*}\right)+\hat{\pi}\right)+\left(1-q\left(n,1\right)\right)\cdot\hat{\pi}-k_{n}\leq\hat{\pi}\,\,\Leftrightarrow\,\,g\left(a^{*}\right)\leq\frac{k_{n}}{q\left(n,1\right)}.
\]
This holds for all $n$ if $g\left(a^{*}\right)\leq c$. Thus, the
probability of deceiving the incumbents is at most $\frac{k_{n}}{g\left(a^{*}\right)}$.
The fact that $g\left(a^{*}\right)\leq c$ implies that the average
payoff of the mutants against the incumbents is less than ${\color{purple}\hat{\pi}+g\left(a^{*}\right)\cdot\frac{k_{n}}{g\left(a^{*}\right)}\leq\,}\hat{\pi}+k_{n}$,
and thus if the mutants are sufficiently rare they are weakly outperformed
(due to paying the cognitive cost of $k_{n}$). We conclude that $\left(\mu,a^{*}\right)$
is an NSC. 
\item \emph{``Only if side.''} Assume that $\left(\mu,a^{\ast}\right)$
is an NSC. Theorem \ref{thm:highest-type-behaivor} implies that $\pi\left(a^{*},a^{*}\right)=\hat{\pi}$.
Assume that $g\left(a^{*}\right)>c.$ The definition of the effective
cost of deception implies that there exists a cognitive level $n$
such that $\frac{k_{n}}{q\left(n,1\right)}<g\left(a^{*}\right)$.
Lemma \ref{Lemma: level 1-1} implies that all the incumbents have
cognitive level 1. Consider mutants with cognitive level $n$ and
completely indifferent preferences (i.e. preferences that induce indifference
between all action profiles). Let $a'$ be a best reply against $a^{*}$.
There is a post-entry focal configuration in which (i) the incumbents
play $a^{*}$ against the mutants, (ii) the mutants play $a'$ when
they deceive an incumbent opponent and $a^{*}$ when they do not deceive
an incumbent opponent, and (iii) the mutants play $a^{*}$ when they
are matched with another mutant. Note that the mutants achieve at
least $\hat{\pi}+g\left(a^{*}\right)\cdot q\left(n,1\right)$ when
they are matched against the incumbents. The gain relative to incumbents,
$g\left(a^{*}\right)\cdot q\left(n,1\right)$, outweighs their additional
cognitive cost of $k_{n}$, by our assumption that $g\left(a^{*}\right)>c.$
Thus the mutants strictly outperform the incumbents. 
\end{enumerate}
\end{proof}

\subsection{Characterisation of NSCs in Generic Games}

In this section we characterise NSCs in generic games, by which we
mean games in which any two different action profiles each give the
same player a different payoff, and each yield a different sum of
payoffs.
\begin{defn}
A (symmetric) game is generic if for each $a,a',b,b'\in A$, $\left\{ a,a'\right\} \neq\left\{ b,b'\right\} $
implies

\[
\pi\left(a,a'\right)\neq\pi\left(b,b'\right),\textrm{\,\ and\,\,}\pi\left(a,a'\right)+\pi\left(a',a\right)\neq\pi\left(b,b'\right)+\pi\left(b',b\right).
\]

For example, if the entries of the payoff matrix $\pi$ are drawn
independently from a continuous distribution on an open subset of
the real numbers, then the induced game is generic with probability
one. 

Note that a generic game admits at most one efficient action profile.
From Corollary \ref{cor:no-efficient-no-NSC} we know that if the
game does not have a symmetric efficient profile then it does not
admit any NSC (and as discussed in Remark \ref{rem:We-note-that},
essentially every interaction admits a symmetric efficient profile).
Hence we can restrict attention to games with exactly one efficient
action profile. Let $\bar{a}$ denote the action played in this unique
profile.
\end{defn}
Next we present our main result: all incumbent types play efficiently
in any NSC of a generic game. 
\begin{thm}
\label{thm-main-result-efficiency} If $\left(\mu^{*},b^{*}\right)$
is an NSC of a generic game with a (unique) efficient outcome $\left(\bar{a},\bar{a}\right)$,
then $b^{*}\equiv\bar{a}$, for all $\theta,\theta'\in C\left(\mu^{*}\right)$;
i.e. all types play the pure action $\bar{a}$ in all matches. 
\end{thm}
\begin{proof}
Assume to the contrary that configuration $\left(\mu^{*},b^{*}\right)$
is an NSC such that there are some $\theta,\theta'\in C\left(\mu^{*}\right)$
such that $b_{\theta}^{N}\left(\theta'\right)\neq\bar{a}$ and $q\left(\theta_{n},\theta_{n'}\right)+q\left(\theta_{n'},\theta_{n}\right)<1$,
or $b_{\theta}^{D}\left(\theta'\right)\neq\bar{a}$ and $q\left(\theta_{n},\theta_{n'}\right)>0$.
Let $\mathring{\theta}$ be the type with the highest cognitive level
among the types that satisfy at least one of the following conditions:

\begin{enumerate}
\item[(A)] $\mathring{\theta}$ plays inefficiently against itself, i.e. $\pi\left(\mathring{\theta},\mathring{\theta}\right)<\hat{\pi}$.
\item[(B)] $\mathring{\theta}$ and an opponent with a weakly higher type play
an inefficient strategy profile, i.e. $0.5\cdot\left(\pi\left(\mathring{\theta},\theta'\right)+\pi\left(\theta',\mathring{\theta}\right)\right)<\hat{\pi}$
for some $\theta'\neq\mathring{\theta}$ with $n_{\mathring{\theta}}\leq n_{\theta'}$.
\item[(C)] A strictly lower type earns strictly more than $\hat{\pi}$ against
$\mathring{\theta}$, i.e. $\pi\left(\theta'{\color{purple}'},\mathring{\theta}\right)>\hat{\pi}$
for some $\theta''\neq\mathring{\theta}$ with $n_{\mathring{\theta}}>n_{\theta''}$.
\end{enumerate}
We will now successively rule out each of these cases. 

Assume first that (A) holds. Let $\hat{u}$ be a utility function
that is identical to $u_{\mathring{\theta}}$ except that: (i) the
payoff of the outcome $\left(\bar{a},\bar{a}\right)$ is increased
by the minimal amount required to make it a best reply to itself,
and (ii) the payoff of some other outcome is altered slightly (to
ensure $\hat{u}$ is not already an incumbent) in a way that does
not change the behaviour of agents. (The formal definition of $\hat{u}$
is provided in Appendix \ref{subsec:Proof-of-Case}.) Suppose that
mutants of type $\hat{\theta}=\left(\hat{u},n_{\theta}\right)$ invade
the population. Consider a focal post-entry configuration in which
the mutants mimic the play of the type $\mathring{\theta}$ incumbents
in all matches except that: (i) the mutants play the efficient profile
$\left(\bar{a},\bar{a}\right)$ among themselves (which yields a higher
payoff than what $\bar{\theta}$ achieves when matched against $\mathring{\theta}$),
and (ii) when the mutants face a higher type they play either $\left(\bar{a},\bar{a}\right)$
or the same deception/Nash equilibrium that the higher types play
against $\bar{\theta}$. It follows that the mutants $\hat{\theta}$
earn a strictly higher payoff than $\mathring{\theta}$ against $\hat{\theta}$,
and a weakly higher fitness than type $\theta$ against all other
types. Thus the mutants strictly outperform the incumbents, which
contradicts the assumption that $\left(\mu^{*},b^{*}\right)$ is an
NSC. The full technical details of this argument are given in Appendix
\ref{subsec:Proof-of-Case}.

Next, assume that case (B) holds and that case (A) does not hold.
This implies that
\[
0.5\cdot\left(\pi\left(\mathring{\theta},\theta'\right)+\pi\left(\theta',\mathring{\theta}\right)\right)<\hat{\pi}=\pi\left(\mathring{\theta},\mathring{\theta}\right)=\pi\left(\theta',\theta'\right).
\]

That is, in the subpopulation that includes types $\mathring{\theta}$
and $\theta'$ the within-type matchings yield higher payoffs than
out-group matchings (an average payoff of less than $\hat{\pi}$).
The following formal argument shows that this property implies dynamic
instability. The fact that $\left(\mu^{*},b^{*}\right)$ is an NSC
implies that $\mu^{*}$ is an NSS in the type game $\Gamma_{\left(\mu^{*},b^{*}\right)}$.
Let $\mathbf{B}$ be the payoff matrix of the type game $\Gamma_{\left(\mu^{*},b^{*}\right)}$
and let $n=\left|C\left(\mu^{*}\right)\right|$. It is well known
(e.g., \citealp[Exercise 6.4.3]{Hofbauer_Sigmund_1988_book}, and
\citealp[ pp. 1--2]{hofbauer2011deterministic}) that an interior
Nash equilibrium of a normal-form game is an NSS if and only if the
payoff matrix is negative semi-definite with respect to the tangent
space, i.e. if and only if $x^{T}\mathbf{B}x\leq0$ for each $x\in\mathbb{R}^{n}$
such that $\sum_{i}x_{i}=0$. Assume without loss of generality that
type $\mathring{\theta}$ ($\theta'$) is represented by the $j$$^{th}$
($k^{th}$) row of the matrix $B$. Let the column vector $x$ be
defined as follows: $x\left(j\right)=1$, $x\left(k\right)=-1$, and
$x\left(i\right)=0$ for each $i\notin\left\{ j,k\right\} $. That
is, the vector $x$ has all entries equal to zero, except for the
$j^{th}$ entry, which is equal to $1$, and the $k^{th}$ entry,
which is equal to $-1$. We have

\begin{eqnarray*}
x^{T}\mathbf{B}x & = & B_{jj}-B_{kj}-B_{jk}+B_{kk}\\
 & = & \pi\left(\bar{a},\bar{a}\right)-k_{n_{\mathring{\theta}}}+\pi\left(\bar{a},\bar{a}\right)-k_{n_{\theta'}}-\left(\pi\left(b_{\mathring{\theta}}\left(\theta'\right),b_{\theta'}\left(\mathring{\theta}\right)\right)-k_{n_{\mathring{\theta}}}+\pi\left(b_{\theta'}\left(\mathring{\theta}\right),b_{\mathring{\theta}}\left(\theta'\right)\right)-k_{n_{\theta'}}\right)\\
 & = & 2\cdot\pi\left(\bar{a},\bar{a}\right)-\left(\pi\left(b_{\mathring{\theta}}\left(\theta'\right),b_{\theta'}\left(\mathring{\theta}\right)\right)+\pi\left(b_{\theta'}\left(\mathring{\theta}\right),b_{\mathring{\theta}}\left(\theta'\right)\right)\right)>0.
\end{eqnarray*}
Thus $\mathbf{B}$ is not negative semi-definite. 

Finally, assume that only case (C) holds. Let $\bar{\theta}$ be an
incumbent type with the highest cognitive level. The fact that case
(B) does not hold implies that $\pi\left(\bar{\theta},\mathring{\theta}\right)=\pi\left(\mathring{\theta},\bar{\theta}\right)=\hat{\pi}$.
The fact that case (C) holds implies that $\pi\left(\theta'',\mathring{\theta}\right)>\hat{\pi}$,
which implies that type $\mathring{\theta}$ has an undominated action
that can yield a deceiving opponent a payoff of more than $\hat{\pi}$
in a deception equilibrium. This contradicts part (2) of Theorem \ref{thm:highest-type-behaivor},
according to which we should have $\left(b_{\bar{\theta}}^{D}\left(\mathring{\theta}\right),b_{\mathring{\theta}}^{D}\left(\bar{\theta}\right)\right)=FMDE\left(\bar{\theta},\mathring{\theta}\right)$.\\
We have shown that no type in the population satisfies either (A),
(B), or (C). The fact that no type satisfies (A) implies that in any
match of agents of the same type both agents play action $\bar{a}$,
and the fact that no type satisfies (B) implies that in any match
between two agents of different types both players play action $\bar{a}$. 
\end{proof}
Combining the results of this section with the above characterisation
of pure NSCs yields the following corollary, which fully characterises
the NSCs of generic games that admit punishment actions (as discussed
in Remark \ref{rem:punishment-actions}, such actions exist in many
economic interactions). The result shows that such games admit an
NSC iff the deviation gain from the pure efficient symmetric profile
is smaller than the effective cost of defection, and when an NSC exists,
its outcome is the pure efficient symmetric profile. In particular,
in any game that admits an efficient symmetric pure Nash equilibrium,
this equilibrium is the unique NSC outcome, and in the Prisoner's
Dilemma mutual cooperation is the unique NSC outcome iff the gain
from defecting against a cooperator is less than the effective cost
of deception, and no NSC exists otherwise.
\begin{cor}
\label{cor:full-characterization-generic-games}Let G be a generic
game that admits a punishment action. The environment admits an NSC
iff there exists an efficient symmetric pure profile $\left(a^{*},a^{*}\right)$
satisfying $g\left(a^{*}\right)\leq c$ (i.e. the deviation gain is
smaller than the effective cost of deception). Moreover, if $\left(\mu,b\right)$
is an NSC, then $b\equiv a^{*}$, and $n=1$ for all $\left(u,n\right)\in C\left(\mu\right)$. 
\end{cor}
\begin{rem}
Corollary \ref{cor:full-characterization-generic-games} shows that
generic games do not admit NSCs if the effective cost of deception
is less than the deviation gain of the efficient profile. In such
cases the distribution of types and their induced behaviour will not
converge to a static population state. We leave the formal analysis
of environments that do not admit NSCs to future research. One conjecture
for the dynamic behaviour in such environments is a never-ending cycle
between states in which almost all agents are naive and play an efficient
action profile, and states in which different cognitive levels coexist,
and agents play inefficient action profiles (see the related analysis
of cyclic behaviour in the Prisoner\textquoteright s Dilemma with
cheap talk and material preferences in \citealp{wiseman2001cooperation}). 
\end{rem}

\begin{rem}
Corollary \ref{cor:full-characterization-generic-games} states that
in an NSC of a generic game everyone has the same cognitive level.
One may wonder how this relates to the apparent cognitive heterogeneity
in the real world. Our analysis in this paper assumes a single underlying
game, while in reality we face a potentially infinite set of games.
If an individual's fitness is the result of interactions in a set
of games that includes generic games with an NSC as well as non-generic
games (see Section \ref{subsec:NSCs-in-non-Generic}) or generic games
that do not admit any NSC (see previous remark), then evolution may
lead to states in which different cognitive levels coexist, possibly
with a never-ending cycle between states with different mixtures of
cognitive levels.
\end{rem}

\begin{rem}
Corollary \ref{cor:full-characterization-generic-games} assumes that
the underlying game admits a punishment action $\tilde{a}$, that
gives an opponent a payoff strictly smaller than the efficient payoff
$\hat{\pi}$, regardless of the opponent's play. This punishment action
is used in the construction of the NSC that induces the efficient
action $a^{*}$. Specifically, a non-deceived incumbent plays the
punishment action $a'$ against any mutant who does not always plays
action $a^{*}$. If the game does not admit a punishment action, then
(1) a complicated game-specific construction of the way in which incumbents
behave against mutants who do not always play $a^{*}$ may be required
to support the efficient action as the outcome of an NSC, and (2)
this construction may require further restrictions on the effective
cost of deception, in addition to $g\left(a^{*}\right)\leq c$. We
leave the study of these issues to future research.\\
\end{rem}

\subsection{Non-Pure NSCs in Non-generic Games\label{subsec:NSCs-in-non-Generic}}

The previous two subsections fully characterise (i) pure NSCs and
(ii) NSCs in generic games. In this section we analyse non-pure NSCs
in non-generic games. Non-generic games may be of interest in various
setups, such as: (1) normal-form representation of generic extensive-form
games (the induced matrix is typically non-generic), and (2) interesting
families of games, such as zero-sum games. Unlike generic games, non-generic
games can admit NSCs that are not pure and that may therefore contain
multiple cognitive levels. To demonstrate this we consider the Rock-Paper-Scissors
game, with the following payoff matrix:\footnote{For the construction presented in this subsection to work, the underlying
game must be non-generic. Observe that if one slightly perturbs the
payoffs of the Rock-Paper-Scissors game to make it a strictly competitive
almost-zero-sum generic game, then Corollary \ref{cor:full-characterization-generic-games}
applies, and the only candidate to be an NSC is a configuration in
which all agents have cognitive level one, and they all play an efficient
action profile.}
\[
\begin{array}{cccc}
 & R & P & S\\
R & 0,0 & -1,1 & 1,-1\\
P & 1,-1 & 0,0 & -1,1\\
S & -1,1 & 1,-1 & 0,0
\end{array}.
\]

To simplify the analysis and the notations we assume in this subsection
that a player always succeeds in deceiving an opponent with a lower
cognitive level, i.e. that $q\left(n,n'\right)=1$ whenever $n{\color{purple}>}n'$.
The analysis can be extended to the more general setup. 

The result below shows that, under mild assumptions on the cognitive
cost function, this game admits an NSC in which all players have the
same materialistic preferences, but players of different cognitive
levels coexist, and non-Nash profiles are played in all matches between
two individuals of different cognitive levels. More precisely, when
individuals of different cognitive levels meet, the higher-level individual
deceives the lower-level individual into taking a pure action that
the higher-level individual then best-replies to. Thus the higher-level
individual earns $1$ and her opponent earns $-1$. Individuals of
the same cognitive level play the unique Nash equilibrium. This means
that higher-level types will obtain a payoff of $1$ more often than
lower-level types, and lower-level types will obtain a payoff of $-1$
more often than higher-level types. In the NSC this payoff difference
is offset exactly by the higher cognitive cost paid by the higher
types. Moreover, the cognitive cost is increasing and unbounded such
that at some point the cost of cognition outweighs any payoff differences
that may arise from the underlying game. This implies that there is
an upper bound on the cognitive sophistication in the population.
\begin{prop}
\label{Prop RPS neutral}Let $G$ be a Rock-Paper-Scissors game. Let
$u^{\pi}$ denote the (materialistic) preference such that $u^{\pi}\left(a,a^{\prime}\right)=\pi\left(a,a^{\prime}\right)$
for all profiles $\left(a,a^{\prime}\right)$. Assume that $q\left(n,n'\right)=1$
whenever $n\neq n'$. Further assume that the marginal cognitive cost
is small but non-vanishing, so that (a) there is an $N$ such that
$k_{N}\leq2<k_{N+1}\text{,}$ and (b) it holds that $1>k_{n+1}-k_{n}\text{ for all }n\leq N$.\textbf{
}Under these assumptions there exists an NSC $\left(\mu^{\ast},b^{\ast}\right)$
such that $C\left(\mu^{\ast}\right)\subseteq\{\left(u^{\pi},n\right)\}_{n=1}^{N}$,
and $\mu^{\ast}$ is mixed (i.e. $\left\vert C\left(\mu^{\ast}\right)\right\vert >1$).
The behaviour of the incumbent types is as follows: if the individuals
in a match are of different cognitive levels, then the higher level
plays Paper and the lower level plays Rock; if both individuals in
a match are of the same cognitive level, then they both play the unique
Nash equilibrium (i.e. \foreignlanguage{british}{randomise} uniformly
over the three actions).
\end{prop}
Appendix C contains a formal proof of this result and relates it to
a similar construction in \citet{Conlisk_2001}. 

Our next result gives a lower bound to the fitness obtained in NSCs.
Let $\underline{M}$ be the pure maxmin value of the underlying game:
\[
\underline{M}=\max_{a_{1}\in A}\min_{a_{2}\in A}\pi\left(a_{1},a_{2}\right).
\]
The pure maxmin value \textit{$\underline{M}$} is the minimal fitness
payoff a player can guarantee herself in the sequential game in which
she plays first, and the opponent replies in an arbitrary way (i.e.
not necessarily maximising the opponent's fitness.)

Proposition \ref{pro:above-min-max-1} shows that the pure maxmin
value is a lower bound on the fitness payoff obtained in an NSC. The
intuition is that if the payoff is lower, then a mutant of cognitive
level $1$, with preferences such that the maxmin action $a_{\text{\textit{\ensuremath{\underline{M}}}}}$
is dominant, will outperform the incumbents.
\begin{prop}
\label{pro:above-min-max-1}If $\left(\mu^{*},b^{*}\right)$ is an
NSC then $\Pi\left(\mu^{*},b^{*}\right)\geq\underline{M}$. 
\end{prop}
\begin{proof}
Assume to the contrary that $\Pi\left(\mu^{*},b^{*}\right)<\underline{M}$.
Let $a_{\text{\textit{\ensuremath{\underline{M}}}}}$ be a maxmin
action of a player, which guarantees that the player's payoff is at
least \textit{\textsubbar{M}}, i.e.\textit{ }$a_{\underline{M}}\in\arg\max_{a_{1}\in A}\min_{a_{2}\in A}\pi\left(a_{1},a_{2}\right).$

Let $u^{a_{\underline{M}}}$ be the preferences in which the player
obtains a payoff of 1 if she plays $a_{\underline{M}}$ and a payoff
of 0 otherwise. Consider a monomorphic group of mutants with type
$\left(u^{a_{\text{\textit{\ensuremath{\underline{M}}}}}},1\right)$.
The fact that $a_{\text{\textit{\ensuremath{\underline{M}}}}}$ is
a maxmin action implies that $\pi_{\left(u^{a_{\underline{M}}},1\right)}\left(\tilde{\mu},\tilde{b}\right)\geq\underline{M}$
in any post-entry configuration. Furthermore, due to continuity it
holds that $\Pi_{\theta}\left(\tilde{\mu},\tilde{b}\right)<\underline{M}$
for any $\theta\in C\left(\mu\right)$ in all sufficiently close focal
post-entry configurations. This contradicts that $\mu^{*}$ is an
NSS in $\Gamma_{\left(\tilde{\mu},\tilde{b}\right)}$, and thus it
contradicts that $\left(\mu^{*},b^{*}\right)$ is an NSC. 
\end{proof}
We conclude by demonstrating that the lower bound of the maxmin payoff
is binding. Specifically, Example \ref{exa:-one} shows an NSC in
a zero-sum game in which the fitness of the incumbents is arbitrarily
close to the lowest feasible payoff in the underlying game -1 (which
is equal to the maxmin payoff).
\begin{example}
\label{exa:-one} Consider the Rock-Paper-Scissors game described
above. Assume that $k_{2}=1$, $k_{3}>2$$,$ and $q\left(2,1\right)=1$.
For each $\epsilon\in\left(0,1\right)$, consider a population in
which $\epsilon$ of the agents have cognitive level 1, and the remaining
$1-\epsilon$ of the agents have level 2. The agents' behaviour is
according to the behaviour described in Proposition \ref{Prop RPS neutral},
i.e.: (1) an agent of level 2 deceives a level-1 opponent into taking
a pure action that the level-2 agent then best-replies to; thus the
level-2 agent earns $1$ and her opponent earns $-1$; and (2) individuals
of the same cognitive level play the unique Nash equilibrium, and
obtain a payoff of zero in the underlying game. When one takes into
account the cognitive cost $k_{2}=1$ of the level-2 agents, this
behaviour implies that all incumbents obtain a fitness of $\epsilon-1$.
An analogous argument to the proof of Proposition \ref{Prop RPS neutral}
implies that this configuration is an NSC.
\end{example}

\section{Extensions\label{sec:Extensions}}

\subsection{Partial Observability When There Is No Deception\label{subsec:Partial-Observability-When}}

As mentioned above, our basic model assumes perfect observability,
and Nash equilibrium behaviour, in matches without deception. In what
follows we briefly describe the results of a robustness check that
relaxes the first of these two assumptions. For brevity, we detail
the full technical analysis in Appendix \ref{sec:partial-observability}.

Specifically, we follow \citet{Dekel_Ely_Yilankaya_2006} and assume
that in matches without deception, each player privately observes
the opponent's type with an exogenous probability $p$, and with the
remaining probability observes an uninformative signal. This general
model extends both our baseline model (where $p=1$) and \citeauthor{Dekel_Ely_Yilankaya_2006}'s
\citeyearpar{Dekel_Ely_Yilankaya_2006} model (which can be viewed
as assuming arbitrarily high deception costs).

The main results of the baseline model ($p=1$) show that (1) only
efficient profiles can be NSCs, and (2) there exist non-Nash efficient
NSCs, provided that the cost of deception is sufficiently large. Our
analysis shows that the former result (namely, stability implies efficiency)
is robust to the introduction of partial observability: (1) a somewhat
weaker notion of efficiency is satisfied by the behaviour of the incumbents
with the highest cognitive level in any NSC for any $p>0$, and (2)
in games such as the Prisoner's Dilemma, we show that only the efficient
profile can be the outcome of an NSC. 

On the other hand, our analysis shows that our second main result
(namely, the stability of non-Nash efficient outcomes) is not robust
to the introduction of partial observability. Specifically, we show
that: (1) non-Nash efficient profiles cannot be NSC outcomes for any
$p<1$ in games like the Prisoner's Dilemma, even when the effective
cost of deception is arbitrarily large; and (2) non-Nash efficient
outcomes cannot be pure NSC outcomes in all games. If a game admits
a profile that is both efficient and Nash, then the profile is an
NSC outcome for any $p\in\left[0,1\right]$. If the underlying game
does not admit such a profile, then our results show that the environment
does not admit a pure NSC for any $p\in\left(0,1\right)$, and that
games like the Prisoner's Dilemma do not admit any NSC. This suggests
that in order to study stability in such environments one might need
to apply weaker solution concepts or to follow a dynamic (rather than
static) approach.

\subsection{Interdependent Preferences\label{subsec:Interdependent-Preferences}}

In the main text we deal exclusively with preferences that are defined
only over action profiles. In what follows we briefly describe how
to extend the analysis to interdependent preferences, i.e. preferences
that may also depend on the opponent's type. A detailed formal analysis
is presented in Appendix \ref{Sect Interdependent}. \citet{Herold_Kuzmics_2009}
study a similar setup while assuming perfect observability of types
among all individuals. Their key result is that any mixed action that
gives each player a payoff above her maxmin payoff can be the outcome
of a stable configuration.\footnote{\citet{Herold_Kuzmics_2009} expand the framework of \citet{Dekel_Ely_Yilankaya_2006}
to include interdependent preferences, i.e. preferences that depend
on the opponent's preference type. Under perfect or almost perfect
observability, if all preferences that depend on the opponent's type
are considered, then any symmetric outcome above the minmax material
payoff is evolutionarily stable. In our setting a pure profile also
has to be a Nash equilibrium in order to be the sole outcome supported
by evolutionarily stable preferences. \citet{Herold_Kuzmics_2009}
find that non-discriminating preferences (including selfish materialistic
preferences) are typically not evolutionarily stable on their own.
By contrast, certain preferences that exhibit discrimination are evolutionarily
stable. Similarly, evolutionary stability requires the presence of
discriminating preferences also in our setup.} 

Our main result for interdependent preferences in our setup shows
that a pure configuration is stable essentially iff: (1) all incumbents
have the same cognitive level $n$, (2) the cost of level $n$ is
smaller than the difference between the incumbents' (fitness) payoff
and the minmax/maxmin value, and (3) the deviation gain is smaller
than the effective cost of deception against an opponent with cognitive
level $n$. In particular, if the marginal effective cost of deception
is sufficiently small, then only Nash equilibria can be the outcomes
of pure stable configurations, while if the effective cost of deceiving
some cognitive level $n$ is sufficiently high (while the cost of
achieving level $n$ is sufficiently low), then essentially any action
profile is the outcome of a pure stable configuration (similar to
the result of \citealp{Herold_Kuzmics_2009}, in the setup without
deception). 

The last part of Appendix \ref{Sect Interdependent} characterises
stable configurations in the Hawk-Dove game. We show that such games
admit heterogeneous stable configurations in which players with different
levels coexist, each type has preferences that induce cooperation
only against itself, and higher types ``exploit'' lower types (and
this is offset by their higher cognitive cost).

\section{Conclusion and Directions for Future Research \label{Sect discussion}}

We have developed a model in which preferences coevolve with the ability
to detect others' preferences and misrepresent one's own preferences.
To this end, we have allowed for heterogeneity with respect to costly
cognitive ability. The assumption of an exogenously given level of
observability of the opponent's preferences, which has characterised
the indirect evolutionary approach up until now, is replaced by the
Machiavellian notion of deception equilibrium, which endogenously
determines what each player observes. Our model assumes a very powerful
form of deception. This allows us to derive sharp results that clearly
demonstrate the effects of endogenising observation and introducing
deception. We think that the ``Bayesian\textquotedblright{} deception
is an interesting model for future research: each incumbent type is
associated with a signal, agents with high cognitive levels can mimic
the signals of types with lower cognitive levels, and agents maximise
their preferences given the received signals and the correct Bayesian
inference about the opponent's type.

In a companion paper (\citealp{HellerMohlin-OoC}) we study environments
in which players are randomly matched, and make inferences about the
opponent's type by observing her past behaviour (rather than directly
observing her type, as is standard in the ``indirect evolutionary
approach\textquotedblright ). In future research, it would be interesting
to combine both approaches and allow the observation of past behaviour
to be influenced by deception.

Most papers taking the indirect evolutionary approach study the stability
of preferences defined over material outcomes. Moreover, it is common
to restrict attention to some parameterised class of such preferences.
Since we study preferences defined on the more abstract level of action
profiles we do not make predictions about whether some particular
kind of preferences over material outcomes, from a particular family
of utility functions, will be stable or not. It would be interesting
to extend our model to such classes of preferences. Furthermore, with
preferences defined over material outcomes it would be possible to
study coevolution of preferences and deception not only in isolated
games, but also when individuals play many different games using the
same preferences. We hope to come back to these questions and we invite
others to employ and modify our framework in these directions.

\appendix

\section{Formal Proofs of Theorems \ref{thm:highest-type-behaivor} and \ref{thm-main-result-efficiency}\label{sec:Formal-Proofs-of}}

\subsection{Preliminaries\label{subsec:Preliminaries-pro-genorus}}

This subsection contains notation and definitions that will be used
in the following proofs.

A generous action is an action such that if played by the opponent,
it allows a player to achieve the maximal fitness payoff. Formally:
\begin{defn}
Action $a_{g}\in A$ is \emph{generous}, if there exists $a\in A$
such that $\pi\left(a,a_{g}\right)\geq\pi\left(a',a''\right)$ for
all $a',a''\in A$.
\end{defn}
Fix a generous action $a_{g}\in A$ of the game $G$. A second-best
generous action is an action such that if played by the opponent,
it allows a player to achieve the fitness payoff that is maximal under
the constraint that the opponent is not allowed to play the generous
action $a_{g}$. Formally: 
\begin{defn}
Action $a_{g_{2}}\in A$ is \emph{second-best generous}, conditional
on $a_{g}\in A$ being first-best generous, if there exists $a\in A$
such that $\pi\left(a,a_{g_{2}}\right)\geq\pi\left(a',a''\right)$
for all $a',a''\in A$ such that $a''\neq a_{g}$.

Fix a generous action $a_{g}\in A$, and fix a second-best generous
action $a_{g_{2}}\in A$, conditional on $a_{g}\in A$ being first-best
generous. For each $\alpha\geq\beta\geq0$, let $u_{\alpha,\beta}$
be the following utility function:
\[
u_{\alpha,\beta}\left(a,a'\right)=\begin{cases}
\alpha & a'=a_{g}\\
\beta & a'=a_{g_{2}}\\
0 & \textrm{otherwise.}
\end{cases}
\]
Observe that such a utility function $u_{\alpha,\beta}$ satisfies: 
\end{defn}
\begin{enumerate}
\item \emph{Indifference:} the utility function only depends on the opponent's
action; i.e. the player is indifferent between any two of her own
actions.
\item \emph{Pro-generosity:} the utility is highest if the opponent plays
the generous action, second-highest if the opponent plays the second-best
generous action, and lowest otherwise. 
\end{enumerate}
Let $U_{GI}=\left\{ u_{\alpha,\beta}|\alpha\geq\beta\geq0\right\} $
be the family of all such preferences, called \emph{pro-generous indifferent
preferences}. Note that $U_{{\color{purple}GI}}$ includes a continuum
of different utilities (under the assumption that $G$ includes at
least three actions). Thus, for any set of incumbent types, we can
always find a utility function in $U_{{\color{purple}GI}}$ that does
not belong to any of the current incumbents.

\subsection{Proof of Theorem \ref{thm:highest-type-behaivor} (Behaviour of the
Highest Types)\label{subsec:Proof-of-Theorem-1}}

\subsubsection{Proof of Theorem \ref{thm:highest-type-behaivor}, Part 1}

Assume to the contrary that $\pi\left(b_{\bar{\theta}}^{N}\left(\bar{\theta}\right),b_{\bar{\theta}}^{N}\left(\bar{\theta}\right)\right)<\hat{\pi}$.
(Note that the definition of $\hat{\pi}$ implies that the opposite
inequality is impossible.) Let $a_{1},a_{2}\in A$ be any two actions
such that $\left(a_{1},a_{2}\right)$ is an efficient action profile,
i.e. $0.5\cdot\left(\pi\left(a_{1},a_{2}\right)+\pi\left(a_{1},a_{2}\right)\right)=\hat{\pi}$.
Let $\theta_{1},\theta_{2},\theta_{3}$ be three types that satisfy
the following conditions: (1) the types are not incumbents: $\theta_{1},\theta_{2},\theta_{3}\notin C\left(\mu^{*}\right)$,
(2) the types have the highest incumbent cognitive level: $n_{\theta_{1}}=n_{\theta_{2}}=n_{\theta_{3}}=\bar{n}$,
and (3) the types have different pro-generosity indifferent preferences;
$u_{\theta_{1}},u_{\theta_{2}},u_{\theta_{3}}\in U_{GI}$ and $u_{\theta_{i}}\neq u_{\theta_{j}}$
for each $i\neq j\in\left\{ 1,2,3\right\} $. Let $\mu'$ be the distribution
that assigns mass $\frac{1}{3}$ to each of these types. The post-entry
type distribution is $\tilde{\mu}=\left(1-\epsilon\right)\cdot\mu^{*}+\epsilon\cdot\mu'$.
Let the post-entry behaviour policy $\tilde{b}$ be defined as follows: 
\begin{enumerate}
\item Behaviour among incumbents respects focality: $\tilde{b}_{\theta}^{N}\left(\theta'\right)=b_{\theta}^{N}\left(\theta'\right)$
and $\tilde{b}_{\theta}^{D}\left(\theta'\right)=b_{\theta}^{D}\left(\theta'\right)$~for
each incumbent pair $\theta,\theta'\in C\left(\mu^{*}\right)$.
\item The mutants play fitness-maximising deception equilibria against incumbents
with lower cognitive levels: $\left(\tilde{b}_{\theta_{i}}^{D}\left(\theta'\right),\tilde{b}_{\theta'}^{D}\left(\theta_{i}\right)\right)\in FMDE\left(\theta_{i},\theta'\right)$
for each $i\in\left\{ 1,2,3\right\} $ and $\theta'\in C\left(\mu^{*}\right)$
with $n_{\theta'}<\bar{n}$. Note that $FMDE\left(\theta_{i},\theta'\right)$
is nonempty in virtue of the construction of $U_{GI}$.
\item In matches without deception between mutants and incumbents, the mutants
mimic $\bar{\theta}$ and the incumbents play the same way they play
against $\bar{\theta}$: $\left(\tilde{b}_{\theta_{i}}^{N}\left(\theta'\right),\tilde{b}_{\theta'}^{N}\left(\theta_{i}\right)\right)=\left(b_{\bar{\theta}}^{N}\left(\theta'\right),b_{\theta'}^{N}\left(\bar{\theta}\right)\right)$,
for each $i\in\left\{ 1,2,3\right\} $ and $\theta'\in C\left(\mu^{*}\right)$. 
\item Two mutants of \textit{different} types play efficiently when meeting
each other: $\tilde{b}_{\theta_{i}}^{N}\left(\theta_{\left(i+1\right)\,\textrm{mod}\,3}\right)=a_{1}$
and $\tilde{b}_{\theta_{i}}^{N}\left(\theta_{\left(i-1\right)\,\textrm{mod}\,3}\right)=a_{2}$
for each $i\in\left\{ 1,2,3\right\} $. 
\item When two mutants of the \textit{same} type meet, they play the same
way $\bar{\theta}$ plays against itself: $\tilde{b}_{\theta_{i}}^{N}\left(\theta_{i}\right)=b_{\bar{\theta}}^{N}\left(\bar{\theta}\right)$
for each $i\in\left\{ 1,2,3\right\} $. 
\end{enumerate}
In virtue of point 1 the construction $\left(\tilde{\mu},\tilde{b}\right)$
is a focal configuration (with respect to $\left(\mu^{*},b^{*}\right)$).
By points 2 and 3 each mutant $\theta_{i}$ earns weakly more than
$\bar{\theta}$ against all incumbent types. By points 4 and 5 each
mutant earns strictly more than $\bar{\theta}$ against the mutants.
In total the average fitness earned by each mutant is strictly higher
than that of $\bar{\theta}$, against a population that follows $\left(\tilde{\mu},\tilde{b}\right)$.
This implies that $\mu'$ is a strictly better reply against $\mu^{*}$
in the population game $\Gamma_{\left(\tilde{\mu},\tilde{b}\right)}$.
Thus, $\mu^{*}$ is not a symmetric Nash equilibrium, and therefore
it is not an NSS, in $\Gamma_{\left(\tilde{\mu},\tilde{b}\right)}$,
which implies that $\mu^{*}$ is not an NSC.

\subsubsection{Proof of Theorem \ref{thm:highest-type-behaivor}, Part 2}

Assume to the contrary that $\left(\left(b_{\bar{\theta}}^{D}\left(\underline{\theta}\right),b_{\underline{\theta}}^{D}\left(\bar{\theta}\right)\right)\right)\not\in FMDE\left(\bar{\theta},\underline{\theta}\right)$.
Let $\hat{\theta}$ be a type that satisfies the conditions of: (1)
not being an incumbent: $\hat{\theta}\notin C\left(\mu^{*}\right)$,
(2) having the highest incumbent cognitive level: $n_{\hat{\theta}}=\bar{n}$,
and (3) having pro-generous indifferent preferences: $u_{\hat{\theta}}\in U_{GI}$.
Let $\mu'$ be the distribution that assigns mass one to type $\hat{\theta}$.
The post-entry type distribution is $\tilde{\mu}=\left(1-\epsilon\right)\cdot\mu^{*}+\epsilon\cdot\mu'$.
Let the post-entry behaviour policy $\tilde{b}$ be defined as follows: 
\begin{enumerate}
\item Behaviour among incumbents respects focality: $\tilde{b}_{\theta}^{N}\left(\theta'\right)=b_{\theta}^{N}\left(\theta'\right)$
and $\tilde{b}_{\theta}^{D}\left(\theta'\right)=b_{\theta}^{D}\left(\theta'\right)$
$\forall\theta,\theta'\in C\left(\mu^{*}\right)$. 
\item In matches with deception between mutants and incumbents , behaviour
is such that the mutants maximise their fitness: $\left(\tilde{b}_{\hat{\theta}}^{D}\left(\theta'\right),\tilde{b}_{\theta'}^{D}\left(\hat{\theta}\right)\right)\in FMDE\left(\hat{\theta},\theta'\right)$
for each $\theta'\in C\left(\mu^{*}\right)$ with $n_{\theta'}<\bar{n}$. 
\item In matches without deception between mutants and incumbents, the mutants
mimic $\bar{\theta}$ and the incumbents play the same way they play
against $\bar{\theta}$: $\left(\tilde{b}_{\hat{\theta}}^{N}\left(\theta'\right),\tilde{b}_{\theta'}^{N}\left(\hat{\theta}\right)\right)=\left(b_{\bar{\theta}}^{N}\left(\theta'\right),b_{\theta'}^{N}\left(\bar{\theta}\right)\right)$,
for each $\theta'\in C\left(\mu^{*}\right)$. 
\item The mutant $\hat{\theta}$ plays against itself the same way $\bar{\theta}$
plays against itself: $\left(\tilde{b}_{\hat{\theta}}^{N}\left(\hat{\theta}\right),\tilde{b}_{\hat{\theta}}^{N}\left(\hat{\theta}\right)\right)=\left(\tilde{b}_{\bar{\theta}}^{N}\left(\bar{\theta}\right),\tilde{b}_{\bar{\theta}}^{N}\left(\bar{\theta}\right)\right)$. 
\end{enumerate}
Note that $\left(\tilde{\mu},\tilde{b}\right)$ is a focal configuration
(with respect to $\left(\mu^{*},b^{*}\right)$) and that $\hat{\theta}$
obtains a strictly higher fitness than $\bar{\theta}$ against a population
that follows $\left(\tilde{\mu},\tilde{b}\right)$. This implies that
$\mu'$ is a strictly better reply against $\mu^{*}$ in the population
game $\Gamma_{\left(\tilde{\mu},\tilde{b}\right)}$. Thus, $\mu^{*}$
is not a symmetric Nash equilibrium, and therefore it is not an NSS,
in $\Gamma_{\left(\tilde{\mu},\tilde{b}\right)}$, which implies that
$\mu^{*}$ is not an NSC.

\subsubsection{Proof of Theorem \ref{thm:highest-type-behaivor}, Part 3}

Assume to the contrary that $\pi\left(\underline{\theta},\bar{\theta}\right)>\hat{\pi}$,
which immediately implies that $\pi\left(\bar{\theta},\underline{\theta}\right)<\hat{\pi}$
and that either $\pi\left(b_{\underline{\theta}}^{|D}\left(\bar{\theta}\right),b_{\bar{\theta}}^{D}\left(\underline{\theta}\right)\right)>\hat{\pi}$
or $\pi\left(b_{\underline{\theta}}^{N}\left(\bar{\theta}\right),b_{\bar{\theta}}^{N}\left(\underline{\theta}\right)\right)>\hat{\pi}$.
Let $\hat{\theta}$ be a type that satisfies the conditions of: (1)
not being an incumbent: $\hat{\theta}\notin C\left(\mu^{*}\right)$,
(2) having the highest incumbent cognitive level: $n_{\hat{\theta}}=\bar{n}$,
and (3) having pro-generous indifferent preferences: $u_{\hat{\theta}}\in U_{GI}$.
Let $\mu'$ be the distribution that assigns mass one to type $\hat{\theta}$.
The post-entry type distribution is $\tilde{\mu}=\left(1-\epsilon\right)\cdot\mu^{*}+\epsilon\cdot\mu'$.
Let the post-entry behaviour policy $\tilde{b}$ be defined as follows: 
\begin{enumerate}
\item Behaviour among incumbents respects focality: $\tilde{b}_{\theta}^{N}\left(\theta'\right)=b_{\theta}^{N}\left(\theta'\right)$
and $\tilde{b}_{\theta}^{D}\left(\theta'\right)=b_{\theta}^{D}\left(\theta'\right)$
$\forall\theta,\theta'\in C\left(\mu^{*}\right)$. 
\item In matches with deception between mutants and incumbents, behaviour
is such that the mutants maximise their fitness: $\left(\tilde{b}_{\hat{\theta}}^{D}\left(\theta'\right),\tilde{b}_{\theta'}^{D}\left(\hat{\theta}\right)\right)\in FMDE\left(\hat{\theta},\theta'\right)$
for each $\theta'\in C\left(\mu^{*}\right)$ with $n_{\theta'}<\bar{n}$. 
\item In a match between a mutant $\hat{\theta}$ and the incumbent $\bar{\theta}$,
the mutant mimics $\underline{\theta}$, and the incumbent $\bar{\theta}$
plays the same way it plays against $\underline{\theta}$: $\left(\tilde{b}_{\hat{\theta}}^{N}\left(\bar{\theta}\right),\tilde{b}_{\bar{\theta}}^{N}\left(\hat{\theta}\right)\right)=\left(b_{\underline{\theta}}^{N}\left(\bar{\theta}\right),b_{\bar{\theta}}^{N}\left(\underline{\theta}\right)\right)$
if $\pi\left(b_{\underline{\theta}}^{N}\left(\bar{\theta}\right),b_{\bar{\theta}}^{N}\left(\underline{\theta}\right)\right)>\hat{\pi}$,
and $\left(\tilde{b}_{\hat{\theta}}^{N}\left(\bar{\theta}\right),\tilde{b}_{\bar{\theta}}^{N}\left(\hat{\theta}\right)\right)=\left(b_{\underline{\theta}}^{D}\left(\bar{\theta}\right),b_{\bar{\theta}}^{D}\left(\underline{\theta}\right)\right)$
otherwise. 
\item The mutant $\hat{\theta}$ plays against itself the same way $\bar{\theta}$
plays against itself: $\left(\tilde{b}_{\hat{\theta}}^{N}\left(\hat{\theta}\right),\tilde{b}_{\hat{\theta}}^{N}\left(\hat{\theta}\right)\right)=\left(\tilde{b}_{\bar{\theta}}^{N}\left(\bar{\theta}\right),\tilde{b}_{\bar{\theta}}^{N}\left(\bar{\theta}\right)\right)$. 
\item The mutant $\hat{\theta}$ mimics $\bar{\theta}$ against all other
incumbents without deception, and these incumbents play against $\hat{\theta}$
in the same way they play against $\bar{\theta}$: $\left(\tilde{b}_{\hat{\theta}}^{N}\left(\theta'\right),\tilde{b}_{\theta'}^{N}\left(\hat{\theta}\right)\right)=\left(b_{\bar{\theta}}^{N}\left(\theta'\right),b_{\theta'}^{N}\left(\bar{\theta}\right)\right)$
for each $\theta'\neq\bar{\theta}$. 
\end{enumerate}
Note that $\left(\tilde{\mu},\tilde{b}\right)$ is a focal configuration
(with respect to $\left(\mu^{*},b^{*}\right)$). By point 2 the mutant
$\hat{\theta}$ earns weakly more than $\bar{\theta}$ against lower
types. By point 3 and Theorem \ref{thm:highest-type-behaivor}.1,
the mutants earn strictly more than $\bar{\theta}$ against type $\bar{\theta}$.
By points 3 and 4 and Theorem \ref{thm:highest-type-behaivor}.1,
the mutant earns strictly more than $\bar{\theta}$ against the mutant.
By point 5 the mutant $\hat{\theta}$ earns the same as $\bar{\theta}$
against all other types. In total the average fitness earned by $\hat{\theta}$
is strictly higher than that of $\bar{\theta}$, against a population
that follows $\left(\tilde{\mu},\tilde{b}\right)$. Recall (Remark
\ref{Remark-internal-stability} in Section \ref{subSect evolution})
that all the incumbent types have the same fitness in $\left(\mu^{*},b^{*}\right)$.
By a standard continuity argument, the fitness of incumbent $\bar{\theta}$
is arbitrarily close (for a sufficiently small $\epsilon$) to the
fitness levels of any other incumbent type in the focal post-entry
configuration $\left(\tilde{\mu},\tilde{b}\right)$. This implies
that $\mu'$ is a strictly better reply against $\mu^{*}$ in the
type game $\Gamma_{\left(\tilde{\mu},\tilde{b}\right)}$. Thus, $\mu^{*}$
is not a symmetric Nash equilibrium, and therefore it is not an NSS,
in $\Gamma_{\left(\tilde{\mu},\tilde{b}\right)}$, which implies that
${\color{purple}\left(\mu^{*},b^{*}\right)}$ is not an NSC.

\subsection{Proof of Case (A) in Theorem \ref{thm-main-result-efficiency}\label{subsec:Proof-of-Case}}

In what follows we fill in the missing technical details for the part
of the proof of Theorem \ref{thm-main-result-efficiency} that concerns
case (A). We begin by proving a lemma.
\begin{lem}
\label{lem: pure selection from DE}If $\left(\sigma_{1},\sigma_{2}\right)\in DE\left(\theta_{1},\theta_{2}\right)$
then there exist actions $a_{1},a_{1}'\in C\left(\sigma_{1}\right)$
and $a_{2},a_{2}'\in C\left(\sigma_{2}\right)$ such that: $\left(a_{1},a_{2}\right)\in DE\left(\theta_{1},\theta_{2}\right)$,
and $\left(a_{1}',a_{2}'\right)\in DE\left(\theta_{1},\theta_{2}\right)$,
with $\pi\left(a_{1},a_{2}\right)\geq\pi\left(\sigma_{1},\sigma_{2}\right)$,
and $\pi\left(a_{1}',a_{2}'\right)\leq\pi\left(\sigma_{1},\sigma_{2}\right)$.

\begin{proof}
Note that for any mixed deception equilibrium $\left(\sigma_{1},\sigma_{2}\right)$
and any action $a\in C\left(\sigma_{2}\right)$, the profile $\left(\sigma_{1},a\right)$
is also a deception equilibrium (because otherwise the deceiver would
not induce the deceived party to take a mixed action that puts positive
weight on $a$). It follows that there are actions $a_{2},a_{2}'\in C\left(\sigma_{2}\right)$
such that $\left(\sigma_{1},a_{2}\right)$ and $\left(\sigma_{1},a_{2}'\right)$
are deception equilibria, with $\pi\left(\sigma_{1},a_{2}\right)\geq\pi\left(\sigma_{1},\sigma_{2}\right)$
and $\pi\left(\sigma_{1},a_{2}'\right)\leq\pi\left(\sigma_{1},\sigma_{2}\right)$.
Furthermore, if $\left(\sigma_{1},a_{2}\right)$ and $\left(\sigma_{1},a_{2}'\right)$
are deception equilibria, then for any action $a\in C\left(\sigma_{1}\right)$,
the profiles $\left(a,a_{2}\right)$ and $\left(a,a'_{2}\right)$
are also deception equilibria, with $\pi\left(\sigma_{1},a_{2}\right)=\pi\left(a,a_{2}\right)$
and $\pi\left(\sigma_{1},a_{2}'\right)=\pi\left(a,a_{2}'\right)$.
Hence there are actions $a_{1},a_{1}'\in C\left(\sigma_{1}\right)$
such that $\left(a_{1},a_{2}\right)$ and $\left(a_{1}',a'_{2}\right)$
are deception equilibria, with $\pi\left(a_{1},a_{2}\right)=\pi\left(\sigma_{1},a_{2}\right)\geq\pi\left(\sigma_{1},\sigma_{2}\right)$,
and $\pi\left(a_{1},a_{2}'\right)=\pi\left(\sigma_{1},a_{2}'\right)\leq\pi\left(\sigma_{1},\sigma_{2}\right)$.
\end{proof}
\end{lem}
Assume that case (A) holds: there is an incumbent $\mathring{\theta}$
that plays inefficiently against itself, i.e. $\left(b_{\mathring{\theta}}^{N}\left(\mathring{\theta}\right),b_{\mathring{\theta}}^{N}\left(\mathring{\theta}\right)\right)$
$\neq\left(\bar{a},\bar{a}\right)$, and there is no incumbent type
with a strictly higher cognitive level than $\mathring{\theta}$ that
satisfies any of the cases (A), (B), or (C). To prove that this cannot
hold in an NSC we introduce a mutant $\hat{\theta}=\left(\hat{u},n_{\mathring{\theta}}\right)\notin C\left(\mu^{*}\right).$
If $\Sigma\left(u_{\mathring{\theta}}\right)=\Delta$, then we let
$\hat{u}\in U_{GI}$ be such that $\hat{\theta}=\left(\hat{u},n_{\mathring{\theta}}\right)\notin C\left(\mu^{*}\right)$.
If $\Sigma\left(u_{\mathring{\theta}}\right)\neq\Delta$, then we
fix a dominated action $\underline{a}\in A\backslash\Sigma\left(u_{\mathring{\theta}}\right)$,
and let $\hat{u}$ be defined as follows:
\[
\hat{u}\left(a,a'\right)=\begin{cases}
max_{a\in A}\left(u_{\mathring{\theta}}\left(a,\bar{a}\right)\right) & a=a'=\bar{a}\\
u_{\mathring{\theta}}\left(\underline{a},a'\right)-\beta_{a'} & a=\underline{a}\,\,and\,\,a'\neq\bar{a}\\
u_{\mathring{\theta}}\left(a,a'\right) & \textrm{otherwise,}
\end{cases}
\]
where each $\beta_{a'}\geq0$ is chosen such that $\hat{\theta}=\left(\hat{u},n_{\mathring{\theta}}\right)\notin C\left(\mu^{*}\right)$.
That is, if $\Sigma\left(u_{\mathring{\theta}}\right)\neq\Delta$,
then the utility function $\hat{u}$ is constructed from the utility
function $u_{\mathring{\theta}}$ by arbitrarily lowering the payoff
of some of the outcomes associated with the (already) dominated action
$\underline{a}$ and that do not involve action $\bar{a}$, while
increasing the payoff of the outcome $\left(\bar{a},\bar{a}\right)$
by the minimal amount that makes $\bar{a}$ a best reply to itself.
Note that this definition of $\hat{u}$ is valid also for the case
of $\bar{a}=\underline{a}$. It follows that $a\in\Sigma\left(u_{\mathring{\theta}}\right)\cup\left\{ \bar{a}\right\} $
iff $a\in\Sigma\left(\hat{u}\right)$. To see this, note that if $\Sigma\left(u_{\mathring{\theta}}\right)\neq\Delta$
and $\underline{a}=\bar{a}$, then $\Sigma\left(\hat{u}\right)=\Sigma\left(u_{\mathring{\theta}}\right)\cup\left\{ \bar{a}\right\} $.
Otherwise $\Sigma\left(\hat{u}\right)=\Sigma\left(u_{\mathring{\theta}}\right)$.
Thus, $\hat{\theta}$ can be induced to play exactly the same pure
actions as $\mathring{\theta}$, unless $\bar{a}=\underline{a}$,
in which case $\hat{\theta}$ can be induced to play $\bar{a}$ in
addition to all actions that $\mathring{\theta}$ can be induced to
play.

Let $\mu'$ be the distribution that assigns mass one to type $\left(\hat{u},n_{\mathring{\theta}}\right)$.
Let the post-entry type distribution be $\tilde{\mu}=\left(1-\epsilon\right)\cdot\mu^{*}+\epsilon\cdot\mu'$,
and let the post-entry behaviour policy $\tilde{b}$ be defined as
follows:
\begin{enumerate}
\item Behaviour among incumbents respects focality: $\tilde{b}_{\theta}^{N}\left(\theta'\right)=b_{\theta}^{N}\left(\theta'\right)$
and $\tilde{b}_{\theta}^{D}\left(\theta'\right)=b_{\theta}^{D}\left(\theta'\right)$
$\forall\theta,\theta'\in C\left(\mu^{*}\right)$.
\item In matches without deception between the mutant type $\hat{\theta}$
and any incumbent type $\theta'$, the mutant $\hat{\theta}$ mimics
$\mathring{\theta}$, and the incumbent $\theta'$ treats the mutant
$\hat{\theta}$ like the incumbent $\mathring{\theta}$: $\left(\tilde{b}_{\hat{\theta}}^{N}\left(\theta'\right),\tilde{b}_{\theta'}^{N}\left(\hat{\theta}\right)\right)=\left(b_{\mathring{\theta}}^{N}\left(\theta'\right),b_{\theta'}^{N}\left(\mathring{\theta}\right)\right)$
for all $\theta'$ such that $n_{\theta'}=n_{\mathring{\theta}}$
and $\theta'\neq\hat{\theta}$. 
\item In matches with deception between the mutant type $\hat{\theta}$
and any lower type $\theta'\in C\left(\mu^{*}\right)$ (with $n_{\theta'}<n_{\hat{\theta}}$),
we distinguish two cases.

\begin{enumerate}
\item Suppose that $\Sigma\left(u_{\mathring{\theta}}\right)=\Delta$. In
this case let $\left(\tilde{b}_{\hat{\theta}}^{D}\left(\theta'\right),\tilde{b}_{\theta'}^{D}\left(\hat{\theta}\right)\right)\in FMDE\left(\hat{\theta},\theta'\right)$.
Note that $FMDE\left(\hat{\theta},\theta'\right)$ is nonempty since
in this case $\hat{u}\in U_{GI}$. 
\item Suppose that $\Sigma\left(u_{\mathring{\theta}}\right)\neq\Delta$.
In this case let $\left(\tilde{b}_{\hat{\theta}}^{D}\left(\theta'\right),\tilde{b}_{\theta'}^{D}\left(\hat{\theta}\right)\right)=\left(a_{1},a_{2}\right)$,
for some $\left(a_{1},a_{2}\right)\in DE\left(\mathring{\theta},\theta'\right)$
such that $\pi\left(a_{1},a_{2}\right)\geq\pi\left(b_{\mathring{\theta}}^{D}\left(\theta'\right),b_{\theta'}^{D}\left(\mathring{\theta}\right)\right)$.
By Lemma \ref{lem: pure selection from DE} above such a profile $\left(a_{1},a_{2}\right)$
exists.
\end{enumerate}
\item The mutant plays efficiently when meeting itself: $\tilde{b}_{\hat{\theta}}^{N}\left(\hat{\theta}\right)=\bar{a}$. 
\item In matches with deception between the mutant $\hat{\theta}$ and a
higher type $\theta'\in C\left(\mu^{*}\right)$ ($n_{\theta'}>n_{\hat{\theta}}$),
we distinguish two cases. Pick a profile $\left(a_{1},a_{2}\right)\in DE\left(\theta',\mathring{\theta}\right)$,
such that $\pi\left(a_{2},a_{1}\right)\geq\pi\left(b_{\mathring{\theta}}^{D}\left(\theta'\right),b_{\theta'}^{D}\left(\mathring{\theta}\right)\right)$.
By Lemma \ref{lem: pure selection from DE} above such a profile $\left(a_{1},a_{2}\right)$
exists. Moreover, by the construction of $\hat{u}$, it is either
the case that $\left(a_{1},a_{2}\right)\in DE\left(\theta',\hat{\theta}\right)$,
or there is some $\tilde{a}$ such that $u_{\theta'}\left(\tilde{a},\bar{a}\right)>u_{\theta'}\left(a_{1},a_{2}\right)$.
In the latter case we have $\left(\bar{a},\bar{a}\right)\in DE\left(\theta',\hat{\theta}\right)$,
due to the fact that $\left(b_{\theta'}^{D}\left(\theta'\right),b_{\theta'}^{D}\left(\theta'\right)\right)=\left(\bar{a},\bar{a}\right)$
implies that $\bar{a}$ is a best reply to $\bar{a}$ for type $\theta'$.

\begin{enumerate}
\item If $u_{\theta'}\left(a_{1},a_{2}\right)>u_{\theta'}\left(\bar{a},\bar{a}\right)$
let $\left(\tilde{b}_{\theta'}^{D}\left(\hat{\theta}\right),\tilde{b}_{\hat{\theta}}^{D}\left(\theta'\right)\right)=\left(a_{1},a_{2}\right)$.
Note that by the definition of $\left(a_{1},a_{2}\right)$ it holds
that $\pi\left(a_{2},a_{1}\right)\geq\pi\left(b_{\mathring{\theta}}^{D}\left(\theta'\right),b_{\theta'}^{D}\left(\mathring{\theta}\right)\right)$.
\item If $u_{\theta'}\left(a_{1},a_{2}\right)\leq u_{\theta'}\left(\bar{a},\bar{a}\right)$
let $\left(\tilde{b}_{\theta'}^{D}\left(\hat{\theta}\right),\tilde{b}_{\hat{\theta}}^{D}\left(\theta'\right)\right)=\left(\bar{a},\bar{a}\right)$.
Note that by the definition of $\mathring{\theta}$ it holds that
$\pi\left(\bar{a},\bar{a}\right)\geq\pi\left(b_{\mathring{\theta}}^{D}\left(\theta'\right),b_{\theta'}^{D}\left(\mathring{\theta}\right)\right)$. 
\end{enumerate}
\end{enumerate}
By point 1, $\left(\tilde{\mu},\tilde{b}\right)$ is a focal configuration
(with respect to $\left(\mu^{*},b^{*}\right)$). By point 2 the mutant
$\hat{\theta}$ earns weakly more than $\mathring{\theta}$ against
lower types. By point 3 the mutant $\hat{\theta}$ earns the same
as $\mathring{\theta}$ against all incumbents of level $n_{\mathring{\theta}}$.
By points 3 and 4 (and the assumption that $\mathring{\theta}$ does
not play efficiently against itself), the mutant $\hat{\theta}$ earns
strictly more than $\mathring{\theta}$ against $\hat{\theta}$. By
point 5 the mutant $\hat{\theta}$ earns weakly more than $\mathring{\theta}$
against all incumbents of a higher cognitive level. In total the average
fitness earned by $\hat{\theta}$ is strictly higher than that of
$\mathring{\theta}$, against a population that follows $\left(\tilde{\mu},\tilde{b}\right)$.
This implies that $\mu'$ is a strictly better reply against $\mu^{*}$
in the population game $\Gamma_{\left(\tilde{\mu},\tilde{b}\right)}$.
Thus, $\mu^{*}$ is not a symmetric Nash equilibrium, and therefore
it is not an NSS of $\Gamma_{\left(\tilde{\mu},\tilde{b}\right)}$,
which implies that $\mu^{*}$ is not an NSC. Thus we have shown that
$\mathring{\theta}$ plays efficiently against itself. 

\section{Type-interdependent Preferences\label{Sect Interdependent}}

As argued by \citet[pp. 542--543]{Herold_Kuzmics_2009}, people playing
a game seem to care not only about the outcome, but also their opponent\textquoteright s
intentions and they discriminate between different types of opponents
(for experimental evidence, see, e.g., \citealp{falk2003nature,charness2007intention}).
Motivated by this observation, in this appendix we extend our baseline
model to allow preferences to depend not only on action profiles,
but also on an opponent's type.

\subsection{Changes to the Baseline Model\label{subsec:Changes-to-the}}

We briefly describe how to extend the model to handle type-interdependent
preferences. Our construction is similar to that of \citet{Herold_Kuzmics_2009}.

When the preferences of a type depend on the opponent's type, we can
no longer work with the set of all possible preferences, because it
would create problems of circularity and cardinality.\footnote{The circularity comes from the fact that each type contains a preferences
component, which is identified with a utility function defined over
types (and action profiles). To see that this creates a problem if
the set of types is unrestricted, let $U_{\ast}$ be the set of all
utility functions that we want to include in our model. Hence $\Theta_{\ast}=U_{\ast}\times\mathbb{N}$
is the set of all types. If $U_{\ast\ast}$ is the set of \emph{all}
mappings $u:A\times A\times\Theta_{\ast}\rightarrow\mathbb{R}$, or,
equivalently, $U_{\ast\ast}$ is the set of \emph{all} mappings $u:A\times A\times U_{\ast}\times\mathbb{N}\rightarrow\mathbb{R}$,
then clearly we have $U_{\ast\ast}\neq U_{\ast}$. See also footnote
10 in \citet{Herold_Kuzmics_2009}.} Instead, we must restrict attention to a pre-specified set of feasible
preferences. We begin by defining $\Theta_{ID}$ as an arbitrary set
of labels. Each label is a pair $\theta=\left(u,n\right)\in\Theta_{ID}$,
where $n\in\mathbb{N}$ and $u$ is a type-interdependent utility
function that depends on the played action profile as well as the
opponent's label, $u:A\times A\times\Theta_{ID}\rightarrow\mathbb{R}$.

Each label $\theta=\left(u,n\right)$ may now be interpreted as a
type. The definition of $u$ extends to mixed actions in the obvious
way. We use the label $u$ also to describe its associated utility
function $u$. Thus $u\left(\sigma,\sigma^{\prime},\theta^{\prime}\right)$
denotes the subjective payoff that a player\ with preferences $u$
earns when she plays strategy $\sigma$ against an opponent with type
$\theta^{\prime}$ who plays strategy $\sigma^{\prime}$.

Let $U_{ID}$ denote the set of all preferences that are part of some
type in $\Theta_{ID}$, i.e. $U_{ID}=\{u:\exists n\in\mathbb{N}$
s.t. $\left(u,n\right)\in\Theta_{ID}\}$. For each preference $\tilde{u}\in U$
of the baseline model (which is defined only over the action profiles)
we can define an equivalent type-interdependent preference $u\in U_{ID}$,
which is independent of the opponent's type; that is, $u\left(\sigma,\sigma^{\prime},\theta^{\prime}\right)=u\left(\sigma,\sigma^{\prime},\theta^{\prime\prime}\right)=\tilde{u}\left(\sigma,\sigma^{\prime}\right)$
for each $\theta^{\prime},\theta^{\prime\prime}\in\Theta_{ID}$ and
$\sigma,\sigma'\in\Delta\left(A\right)$. Let $U_{N}$ denote the
set of all such type-interdependent versions of the preferences of
the baseline model. To simplify the statements of the results of Section
\ref{subsec:Characterisation-of-Pure-interdepen}, in what follows
we assume that $U_{N}\subseteq U_{ID}$.

Next, we amend the definitions of Nash equilibrium, undominated strategies,
and deception equilibrium. The best-reply correspondence now takes
both strategies and types as arguments: $BR_{u}\left(\sigma^{\prime},\theta^{\prime}\right)=\arg\max_{\sigma\in\Delta\left(A\right)}u\left(\sigma,\sigma^{\prime},\theta^{\prime}\right)$.
Accordingly we adjust the definition of the set of Nash equilibria,
\[
NE\left(\theta,\theta^{\prime}\right)=\left\{ \left(\sigma,\sigma^{\prime}\right)\in\Delta\left(A\right)\times\Delta\left(A\right):\sigma\in BR_{u}\left(\sigma^{\prime},\theta^{\prime}\right)\text{ and }\sigma^{\prime}\in BR_{u^{\prime}}\left(\sigma,\theta\right)\right\} ,
\]
and the set of \emph{undominated strategies,}
\[
\Sigma\left(\theta\right)=\left\{ \sigma\in\Delta\left(A\right):\text{there exists }\sigma^{\prime}\in\Delta\left(A\right)\text{ and }\theta^{\prime}\in\Theta_{ID}\text{ such that }\sigma\in BR_{u}\left(\sigma^{\prime},\theta^{\prime}\right)\right\} .
\]
Finally, we adapt the definition of deception equilibrium. Given two
types $\theta,\theta^{\prime}$ with $n_{\theta}>n_{\theta^{\prime}},$
a strategy profile $\left(\tilde{\sigma},\tilde{\sigma}^{\prime}\right)$
is a \emph{deception equilibrium }if\emph{ }
\[
\left(\tilde{\sigma},\tilde{\sigma}^{\prime}\right)\in\arg\max_{\sigma\in\Delta\left(A\right),\sigma^{\prime}\in\Sigma\left(\theta^{\prime}\right)}u_{\theta}\left(\sigma,\sigma^{\prime},\theta^{\prime}\right).
\]
The interpretation of this definition is that the deceiver is able
to induce both a belief about the deceiver's preferences, and a belief
the deceiver's intention, in the mind of the deceived party. Let $DE\left(\theta,\theta^{\prime}\right)$
be the set of all such deception equilibria. The rest of our model
remains unchanged.

Some of the following results rely on the existence of preferences
$u_{\tilde{a}^{\prime},\tilde{n}}^{_{\tilde{a}}}$ that satisfy two
conditions: (1) action $\tilde{a}$ is a (subjective) dominant action
against an opponent with the same preferences and with cognitive level
$\tilde{n}$, and (2) action $\tilde{a}^{\prime}$ is the dominant
action against all other opponents. Formally:
\begin{defn}
Given any two actions $\tilde{a},\tilde{a}^{\prime}\in A,$ let $u_{\tilde{a}^{\prime},\tilde{n}}^{_{\tilde{a}}}$\ be
the discriminating preferences defined by the following utility function:
for all $a,a^{\prime}\in A$ and $\theta'\in U_{ID}$,
\[
u_{\tilde{a}^{\prime},\tilde{n}}^{_{\tilde{a}}}\left(a,a^{\prime},\theta^{\prime}\right)=\left\{ \begin{array}{cc}
1 & \left(\theta^{\prime}=\left(u_{\tilde{a}^{\prime},\tilde{n}}^{_{\tilde{a}}},\tilde{n}\right)\,\text{and}\,a=\tilde{a}\right)\,or\,\left(\theta^{\prime}\neq\left(u_{\tilde{a}^{\prime},\tilde{n}}^{_{\tilde{a}}},\tilde{n}\right)\ \text{and }a=\tilde{a}^{\prime}\right)\\
0 & \text{otherwise}.
\end{array}\right.
\]

Finally, define the \emph{effective cost of deceiving cognitive level}
$n$, denoted by $c\left(n\right)$, as the minimal ratio between
the additional cognitive cost and the probability of deceiving an
opponent of cognitive level $n$: 
\[
c\left(n\right)=\min_{m>n}\,\,\frac{k_{m}-k_{n}}{q\left(m,n\right)}.
\]
Note that $c\left(1\right)\equiv c$, which coheres with the definition
of the effective cost of deception\emph{ }(with respect to cognitive
level $1$) in the baseline model. 
\end{defn}

\subsection{Pure Maxmin and Minimal Fitness\label{subsec:Pure-Maxmin-and}}

The pure maxmin and minmax values give a minimal bound to the fitness
of an NSC. Given a game $G=\left(A,\pi\right),$ define $\underline{M}$
and $\bar{M}$ as its pure maxmin and minmax values, respectively:
\[
\underline{M}=\max_{a_{1}\in A}\min_{a_{2}\in A}\pi\left(a_{1},a_{2}\right),\,\,\,\,\,\,\,\,\,\,\,\,\,\,\overline{M}=\min_{a_{2}\in A}\max_{a_{1}\in A}\pi\left(a_{1},a_{2}\right).
\]
The pure maxmin value \textit{$\underline{M}$} is the minimal fitness
payoff a player can guarantee herself in the sequential game in which
she plays first, and the opponent replies in an arbitrary way. The
pure minmax value $\overline{M}$ is the minimal fitness payoff a
player can guarantee herself in the sequential game in which her opponent
plays first an arbitrary action, and she best-replies to the opponent's
pure action. It is immediate that $\underline{M}\leq\overline{M}$
and that the minmax value in mixed actions is between these two values.

Let $a_{\text{\textit{\b{M}}}}$ be a maxmin action of a player; i.e.
an action $a_{\text{\textit{\b{M}}}}$ guarantees that the player's
payoff is at least \textit{\textsubbar{M},} and let $a_{\bar{M}}$
be a minmax action, i.e. an action that guarantees that the opponent's
payoff is at most $\bar{M}$:\textit{ }
\[
a_{\underline{M}}\in\arg\max_{a_{1}\in A}\min_{a_{2}\in A}\pi\left(a_{1},a_{2}\right),\,\,\,\,\,\,\,\,\,\,\,a_{\bar{M}}\in\arg\min_{a_{2}\in A}\max_{a_{1}\in A}\pi\left(a_{1},a_{2}\right).
\]

The proof of Proposition \ref{pro:above-min-max-1} holds with minor
changes also in the setup of interdependent preferences (under the
assumption that $\left(u^{a_{\underline{M}}},1\right)\in\Theta_{ID}$),
and this implies that the maxmin value is a lower bound on the fitness
payoff obtained in an NSC (i.e. if $\left(\mu,b\right)$ is an NSC
then $\Pi\left(\mu,b\right)\geq\underline{M}$ ).

\subsection{Characterisation of Pure Stable Configurations\label{subsec:Characterisation-of-Pure-interdepen}}

In this subsection we show that, essentially, a pure configuration
is stable if and only if (1) all incumbents have the same cognitive
level $n$, (2) the cost of level $n$ is smaller than the difference
between the incumbents' (fitness) payoff and the minmax/maxmin values,
and (3) the deviation gain is smaller than the effective cost of deceiving
cognitive level $n$. 

We begin by formally stating and proving the necessity claim. 
\begin{prop}
If $\left(\mu^{*},a^{\ast}\right)$ is a pure NSC then the following
holds: (1) if $\theta,\theta'\in C\left(\mu^{*}\right)$ then $n_{\theta}=n_{\theta'}=n$
for some $n$, (2) $\pi\left(a^{\ast},a^{\ast}\right)-\underline{M}\geq k_{n}$,
and (3) $g\left(a^{\ast}\right)\leq c\left(n\right)$.
\end{prop}
\begin{proof}
~

\begin{enumerate}
\item Since all players earn the same game payoff of $\pi\left(a^{\ast},a^{\ast}\right),$
they must also incur the same cognitive cost, or else the fitness
of the different incumbent types would not be balanced (which would
contradict the fact that $\left(\mu,a^{\ast}\right)$ is an NSC). 
\item Assume to the contrary that $\pi\left(a^{\ast},a^{\ast}\right)-\underline{M}<k_{n}$.
A mutant of type $\left(\pi,1\right)$ will be able to earn at least
$\underline{M}$ against incumbents in any post-entry focal configuration.
As the fraction of mutants vanishes the average fitness of mutants
is weakly higher than $\underline{M}$, whereas the fitness of the
incumbents converges to $\pi\left(a^{\ast},a^{\ast}\right)-k_{n}$.
Thus, if it were the case that $\pi\left(a^{\ast},a^{\ast}\right)-\underline{M}<k_{n}$,
then the mutants would outperform the incumbents. 
\item Assume to the contrary that $g\left(a^{\ast}\right)>c\left(n\right)$.
This implies that there exists a cognitive level $m>n$ such that
$g\left(a^{\ast}\right)>\frac{k_{m}-k_{n}}{q\left(m,n\right)}$. Let
$\tilde{a}$ be the fitness best reply against $a^{*}$. Let $\tilde{u}\in U_{N}$
be the preferences that assign a subjective payoff of one if the agent
plays either $\tilde{a}$ or $a^{*}$ and the opponent plays $a^{*}$,
and zero otherwise, i.e. $\tilde{u}\left(a,a',\theta'\right)=\mathbf{1}_{a\in\left\{ a^{*},\tilde{a}\right\} \,and\,a'=a^{*}}$.
There is a focal post-entry configuration in which all agents play
action $a^{*}$ in all interactions except when a deceiving mutant
plays action $\tilde{a}$. A mutant of type $\left(\tilde{u},m\right)$
will then earn $\pi\left(a^{\ast},a^{\ast}\right)+g\left(a^{\ast}\right)\cdot q\left(m,n\right)$
against the incumbents. As the fraction of mutants vanishes the average
fitness of mutants is weakly higher than 
\[
\pi\left(a^{\ast},a^{\ast}\right)+g\left(a^{\ast}\right)\cdot q\left(m,n\right)-k_{m}>\pi\left(a^{\ast},a^{\ast}\right)+\left(k_{m}-k_{n}\right)-k_{m}=\pi\left(a^{\ast},a^{\ast}\right)-k_{n},
\]
whereas the fitness of the incumbents is weakly below $\pi\left(a^{\ast},a^{\ast}\right)-k_{n}$.
Thus, if it were true that $g\left(a^{\ast}\right)>c\left(n\right)$,
the mutants would strictly outperform the incumbents.
\end{enumerate}
\end{proof}
Next, we state and prove the sufficiency claim.
\begin{prop}
\label{Prop: high-cost-implies-stability-1}Suppose \textup{that}
$\hat{\theta}:=\left(u_{a_{\bar{M},n}}^{_{a^{\ast}}},n\right)\in\Theta_{ID}$.
If $\pi\left(a^{\ast},a^{\ast}\right)-\overline{M}>k_{n}$, and $g\left(a^{\ast}\right)<c\left(n\right)$,
then $\left(\hat{\theta},a^{\ast}\right)$ is an ESC. 
\end{prop}
\begin{proof}
Suppose that all incumbents are of type $\left(u_{a_{\bar{M,n}}}^{_{a^{\ast}}},n\right)$.
Note that in all focal post-entry configurations the incumbent $\hat{\theta}$
always plays either $a^{\ast}$ or $a_{\bar{M}}$. Moreover, whenever
an incumbent agent is non-deceived, then she plays action $a^{*}$
against a fellow incumbent and action $a_{\bar{M}}$ against a mutant.
The fact that $\pi\left(a^{\ast},a^{\ast}\right)-k_{n}>\overline{M}$
implies that any mutant $\theta\neq\hat{\theta}$ with cognitive level
$n_{\theta'}\leq n$ earns a strictly lower payoff against the incumbents
in any focal post-entry configuration. As a result, if the frequency
of mutants is sufficiently small, then they are strictly outperformed.
Against a mutant $\left(\theta',n'\right)$ with cognitive level $n'>n$,
an incumbent may play action $a^{\ast}$ only when she is being deceived.
Since $\pi\left(a^{\ast},a^{\ast}\right)>\overline{M}$ the mutants
earn (on average) at most $\pi\left(a^{\ast},a^{\ast}\right)+g\left(a^{\ast}\right)\cdot q\left(n',n\right)$
in matches against incumbents. Consequently, as the fraction of mutants
vanishes, the average fitness of mutants is weakly less than 
\[
\pi\left(a^{\ast},a^{\ast}\right)+g\left(a^{\ast}\right)\cdot q\left(n',n\right)-k_{n'}<\pi\left(a^{\ast},a^{\ast}\right)+\frac{k_{n'}-k_{n}}{q\left(n',n\right)}\cdot q\left(n',n\right)-k_{n'}=\pi\left(a^{\ast},a^{\ast}\right)-k_{n},
\]
and the average fitness of the incumbents converges to $\pi\left(a^{\ast},a^{\ast}\right)-k_{n}$.
Hence, the mutants are outperformed. 
\end{proof}
In particular, our results imply that:
\begin{enumerate}
\item Any pure equilibrium that induces a payoff above the minmax value
$\overline{M}$ is the outcome of a pure ESC (regardless of the cost
of deception). 
\item If the effective cost of deception is sufficiently small, then only
Nash equilibria can be the outcomes of pure NSCs. Specifically, this
is the case if $c\left(n\right)<g\left(a\right)$ for each cognitive
level $n$ and each action $a$ such that $\left(a,a\right)$ is not
a Nash equilibrium of the fitness game.
\item If there is a cognitive level $n$, such that (1) the cost of achieving
level $n$ is sufficiently small, and (2) the effective cost of deceiving
an opponent of level $n$ is sufficiently high, then essentially any
pure profile is the outcome of a pure ESC (similar to the results
of \citealp{Herold_Kuzmics_2009}, in the setup without deception).
Formally, let $A'\subseteq A$ be the set of actions that induce a
payoff above the minmax value: $A'=\left\{ a\in A|\pi\left(a,a\right)>\bar{M}\right\} $.
Assume that there is a cognitive level $n$, such that (1) $k_{n}<\pi\left(a,a\right)-\bar{M}$
for each action $a\in A'$ and (2) $c\left(n\right)>g\left(a\right)$
for each action $a$. Then any action $a\in A'$ is the outcome of
a pure ESC (in which all incumbents have cognitive level $n$).
\end{enumerate}

\subsection{Application: In-group Cooperation and Out-group Exploitation\label{subsec:Application-Hawk-Dove}}

The following table represents a family of Hawk-Dove games. When both
players play $D$ (Dove) they earn $1$ each and when they both play
$H$ (Hawk) they earn $0$. When a player plays $H$ against an opponent
playing $D$, she obtains an additional gain of $g>0$ and the opponent
incurs a loss of $l\in\left(0,1\right)$.
\begin{equation}
\begin{array}{ccc}
 & H & D\\
H & 0,0 & 1+g,1-l\\
D & 1-l,1+g & 1,1
\end{array}.\label{Matrix PD/HD}
\end{equation}
It is natural to think of a mutual play of $D$ as the cooperative
outcome. We define preferences that induce players to cooperate with
their own kind and to seek to exploit those who are not of their own
kind.
\begin{defn}
Let $u^{n}$ denote the preferences such that:

\begin{enumerate}
\item If $u_{\theta^{\prime}}=u^{n}$ and $n_{\theta^{\prime}}=n$ then
$u^{n}\left(D,a^{\prime},\theta^{\prime}\right)=1$ and $u^{n}\left(H,a^{\prime},\theta^{\prime}\right)=0$
for all $a^{\prime}$.
\item If $u_{\theta^{\prime}}\neq u^{n}$ or $n_{\theta^{\prime}}\neq n$
then $u^{n}\left(H,a^{\prime},\theta^{\prime}\right)=1$ and $u^{n}\left(D,a^{\prime},\theta^{\prime}\right)=0$
for all $a^{\prime}$. 
\end{enumerate}
\end{defn}
Thus, when facing someone who is of the same type, an individual with
$u^{n}$-preferences strictly prefers cooperation, in the sense of
playing $D$. When facing someone who is not of the same type, an
individual with $u^{n}$-preferences strictly prefers the aggressive
action $H$. 

To simplify the analysis and the notation in this example we assume\emph{
}that a player always succeeds in deceiving an opponent with a lower
cognitive level; i.e. we assume that $q\left(n,n'\right)=1$ whenever
$n{\color{purple}>}n'$. 

Under the assumption that $g>l$ and that the marginal cognitive costs
are sufficiently small (but non-vanishing), we construct an ESC in
which only individuals with preferences from $\{u^{i}\}_{i=1}^{\infty}$
are present. Individuals of different cognitive levels coexist, and
non-Nash profiles are played in all matches between equals. When individuals
of the same level meet, they play mutual cooperation $\left(D,D\right)$.
When individuals of different levels meet, the higher level plays
$H$ and the lower level plays $D$. The gain from obtaining the high
payoff of $1+g$ against lower types is exactly counterbalanced by
the higher cognitive costs. By contrast, if $g<l$ then the game does
not admit this kind of stable configuration. 
\begin{prop}
\label{Prop PD/HD heterogeneity}Let $G$ be the game represented
in (\ref{Matrix PD/HD}), where $g>0$ and $l\in\left(0,1\right)$.
Assume that $q\left(n,n'\right)=1$ whenever $n\neq n'$. Suppose
that the marginal cognitive cost is small but non-vanishing, so that
(a) there is an $N$ such that $k_{N}\leq l+g<k_{N+1}$, and (b) it
holds that $g>k_{n+1}-k_{n}\text{ for all }n\leq N$.

(i) If $g>l$ then there exists an ESC $\left(\mu^{\ast},b^{\ast}\right)$,
such that $C\left(\mu^{\ast}\right)\subseteq\{\left(u^{n},n\right)\}_{n=1}^{N}$,
and $\mu^{\ast}$ is mixed (i.e. $\left\vert C\left(\mu^{\ast}\right)\right\vert >1$).
The behaviour of the incumbents is as follows: if the individuals
in a match are of different cognitive levels, then the higher level
plays H and the lower level plays D; if both individuals in a match
are of the same cognitive level, then they both play D.

(ii) If $g=l$ then there exists an NSC with the above properties.

(iii) If $g<l$ then there does not exist any NSC $\left(\mu^{\ast},b^{\ast}\right)$,
such that $C\left(\mu^{\ast}\right)\subseteq\{\left(u^{n},n\right)\}_{n=1}^{\infty}$. 
\end{prop}
The formal proof is presented in Appendix C.
\begin{rem}
It is possible to construct an ESC that is like Proposition \ref{Prop PD/HD heterogeneity}(i)
except that when incumbents of the same cognitive level meet they
play the mixed equilibrium of the Hawk-Dove game. Thus we can have
ESCs in which agents mix at the individual level. For instance, this
can be accomplished by considering preferences $u^{m}$ such that:
(1) if $u_{\theta^{\prime}}=u^{m}$ and $n_{\theta^{\prime}}=n$ then
$u^{m}\left(a,a^{\prime},\theta^{\prime}\right)=\pi\left(a,a^{\prime},\theta^{\prime}\right)$
for all $a$ and $a^{\prime}$, and (2) if $u_{\theta^{\prime}}\neq u^{m}$
or $n_{\theta^{\prime}}\neq n$ then $u^{n}\left(H,a^{\prime},\theta^{\prime}\right)=1$
and $u^{n}\left(D,a^{\prime},\theta^{\prime}\right)=0$ for all $a^{\prime}$.
\end{rem}

\section{Constructions of Heterogeneous NSCs in Examples\label{sec:Constructions-of-Heterogeneous}}

Appendix C appears in the supplementary material that can be found
online.

\section{\label{sec:partial-observability}Partial Observability When There
Is No Deception}

Appendix D appears in the supplementary material that can be found
online.

\bibliographystyle{econometrica}
\phantomsection\addcontentsline{toc}{section}{\refname}\bibliography{references}

\end{document}